\documentclass[12pt]{article}

\usepackage[T1]{fontenc}
\usepackage[utf8]{inputenc}
\usepackage{lmodern}

\usepackage[left=.8in,right=.8in,top=.8in,bottom=.8in,footskip=.3in]{geometry}

\usepackage{amsmath,amssymb,amsfonts,amsthm}
\usepackage{bm}            
\allowdisplaybreaks

\usepackage{setspace}
\usepackage{graphicx}
\usepackage{subcaption}
\graphicspath{{./Examples/}}

\usepackage{algorithm}
\usepackage{algpseudocode}

\usepackage{placeins}
\usepackage{arydshln}

\usepackage{xcolor}
\usepackage[
  colorlinks=true,
  linkcolor=blue,
  citecolor=blue,
  urlcolor=blue
]{hyperref}

\usepackage{tikz}
\usetikzlibrary{arrows,automata,calc,graphs,graphs.standard,patterns}
\usepackage{tikz-qtree}

\usepackage{pifont}

\newtheorem{theorem}{Theorem}[section]

\newtheorem{definition}{Definition}[section]
\newtheorem{example}{Example}[section]

\newtheorem{remark}{Remark}[section]

\DeclareMathOperator{\bzero}{\pmb{0}}
\DeclareMathOperator{\bt}{\pmb{\theta}}
\DeclareMathOperator{\mC}{\mathcal{C}}
\DeclareMathOperator{\mL}{\mathcal{L}}

\DeclareMathOperator{\bSig}{\pmb{\Sigma}}
\DeclareMathOperator{\bI}{\pmb{I}}

\title{Safe hypotheses testing with application to \\ order restricted inference}

\author{Ori Davidov}

\date{Department of Statistics, University of Haifa, Mount Carmel, Haifa 3498838 Israel\\
\bigskip
E-mail: \texttt{davidov@stat.haifa.ac.il}
}

\begin{document}

\maketitle

\begin{abstract}
Hypothesis tests under order restrictions arise in a wide range of scientific applications. By exploiting inequality constraints, such tests can achieve substantial gains in power and interpretability.
However, these gains come at a cost: when the imposed constraints are misspecified, the resulting inferences may be misleading or even invalid, and Type III errors may occur, i.e., the null hypothesis may be rejected when neither the null nor the alternative is true. To address this problem, this paper introduces safe tests. Heuristically, a safe test is a testing procedure that is asymptotically free of Type III errors. The proposed test is accompanied by a certificate of validity, a pre--test that assesses whether the original hypotheses are consistent with the data, thereby ensuring that the null hypothesis is rejected only when warranted, enabling principled inference without risk of systematic error. Although the development in this paper focus on testing problems in order--restricted inference, the underlying ideas are more broadly applicable. The proposed methodology is evaluated through simulation studies and the analysis of well--known illustrative data examples, demonstrating strong protection against Type III errors while maintaining power comparable to standard procedures. 

\medskip
   

\end{abstract}

\section{Introduction} \label{section:introduction}

Hypothesis testing has been studied extensively within the framework of order--restricted inference (ORI); see the monographs of Barlow et al. (1972), Robertson et al. (1988), and Silvapulle and Sen (2005). Silvapulle and Sen (2005) classified a large subset of the testing problems arising in ORI into Type A or Type B Problems. Type A Problems are formulated as
\begin{equation} \label{Eq.Hypotheses.TypeA}
H_{0}:\boldsymbol{\theta }\in \mathcal{L} \text{ versus } H_{1}: \boldsymbol{\theta }\in \mathcal{C}\backslash \mathcal{L}
\end{equation}
where $\mathcal{L}$ is a linear subspace and $\mathcal{C}$ is closed convex cone with $\mathcal{L}\subset \mathcal{C}$. A classic example of such testing problems is $H_{0}:\boldsymbol{\theta }=\boldsymbol{0}$ versus $H_{1}:\boldsymbol{\theta}\in \mathbb{R}_{+}^{m}\backslash \{\boldsymbol{0}\}$ where $\mathbb{R}_{+}^{m}$ is the positive orthant. Type A problems are common in applications and often referred to as testing for an order. Type B Problems are formulated as 
\begin{equation} \label{Eq.Hypotheses.TypeB}
H_{0}:\bt\in \mathcal{C} \text{ versus }H_{1}:\bt\notin \mathcal{C}.  
\end{equation}
A canonical Type B Problem is $H_{0}:\boldsymbol{\theta }\in \mathbb{R}_{+}^{m}$ versus $H_{1}:\boldsymbol{\theta }\notin \mathbb{R}_{+}^{m}$.  
This class of tests is referred to as testing against an order. 

It is well known that accounting for constraints, as specified in \eqref{Eq.Hypotheses.TypeA} or \eqref{Eq.Hypotheses.TypeB}, improves the power of the resulting tests (e.g.,Praestgaard 2012) as well as the accuracy of the associated estimators (e.g., Hwang and Peddada 1994, Silvapulle and Sen 2005, Rosen and Davidov 2017). These improvements are often substantial (cf., Singh et al. 2021). A case in point is ANOVA type problems where the superior performance of ORI has been well known for over fifty years, cf., Barlow et al. (1972). Singh and Davidov (2019) recently showed that striking gains are possible when experiments are both designed and analyzed using methods that properly account for the underlying constraints. However, to date, few scientific studies have capitalized on these findings, and the methods of ORI remain vastly underutilized. In our view, barriers to the broad adoption of the methods of ORI are both practical and principled. Practically, ORI requires constrained estimation and nonstandard asymptotic theory, making it more complex to understand and implement. Moreover, standard tools such as the bootstrap may fail when parameters lie on the boundary of the parameter space (Andrews, 2000), and user–friendly software remains limited. Principled objections concern the behavior of tests and estimators when the assumed constraints, e.g., $\boldsymbol{\theta} \in \mathcal{C}$ are misspecified. For example, one may ask how a test for \eqref{Eq.Hypotheses.TypeA} behaves when $\boldsymbol{\theta} \notin \mathcal{C}$. Additionally, several authors, including Silvapulle (1997) and Cohen and Sackrowitz (2004), have discussed methodological concerns and potential deficiencies of the likelihood ratio test (LRT) in ORI. See also Perlman and Wu (1999) and the references therein. Such concerns have motivated the development of alternative procedures, including cone--order monotone tests as advocated by Cohen and Sackrowitz (1998). This communication addresses the aforementioned principled concerns, thereby resolving many of the issues raised in the literature.

Despite the well--known possibility of misspecifying ordered restrictions and the widely recognized risk of Type III errors, there is a clear gap in the literature concerning their formal treatment. Addressing this gap, the paper introduces and studies a novel, easy--to--apply safe test, a testing procedure that is asymptotically free of Type III errors. Safe tests constitute a first step toward adaptive ORI, methodologies for estimation, prediction, and related tasks, in which order constraints are imposed only when supported by the data.

The paper is organized in the following way. In Section \ref{Sec.Geom} the geometry of the distance test is studied. Section \ref{Sec.SafeTest} introduces and studies a novel safe test. Simulation results and illustrative examples, including the reanalysis of some well known case studies from the literature, are provided in Section \ref{Sec.Numerics}. We conclude in Section \ref{Sec.Summary} with a brief summary and a discussion. All proofs are collected in Appendix A.

\section{The distance test and Type III errors} \label{Sec.Geom}

Suppose that there exists a statistic $\pmb{S}_{n}$ which estimates a parameter $\bt \in \Theta \subseteq \mathbb{R}^{m}$ and satisfies 
\begin{equation} \label{Eq.Sn}
\sqrt{n}(\pmb{S}_{n}-\bt)\Rightarrow \mathcal{N}
_{m}(\bzero,\bSig)  
\end{equation}
as $n\rightarrow \infty $ where $\Rightarrow $ denotes convergence in distribution. We further assume that $\bSig_{n},$ a consistent estimator for $\bSig$, exists. Numerous tests for \eqref{Eq.Hypotheses.TypeA} and \eqref{Eq.Hypotheses.TypeB} assuming \eqref{Eq.Sn}  have been proposed in the literature (Silvapulle and Sen, 2005). The most common in both applications as well in theoretical studies is the distance test (DT) which is of the form
\begin{equation} \label{Eq.DT}
T_{n}=T_{n}(\Theta_0,\Theta_1)=n\{\|\pmb{S}_{n}-\Pi_{{\bSig}_{n}}(\pmb{S}_{n}|\Theta_0)\|_{\bSig_{n}}^{2}- \| \pmb{S}_{n}-\Pi_{\bSig_{n}}(\pmb{S}_{n}|\Theta_1)\|_{\bSig_{n}}^{2}\},  
\end{equation}
where $(\Theta_0,\Theta_1)=(\mL,\mC)$ for Type A Problems and $(\Theta_0,\Theta_1)=(\mC,\mathbb{R}^m)$ for Type B Problems. Here $\Pi_{\bSig_{n}}(\pmb{S}_{n}|\Theta_{i})$ is the $\bSig_{n}-$projection of $\pmb{S}_{n}$ onto $\Theta_{i}$ where $i\in \{0,1\}$ and $\|\boldsymbol{\cdot}\|_{\bSig_{n}}$ is the corresponding norm. The DT is the large sample version of the LRT under normality, i.e., if \eqref{Eq.Sn} holds exactly and $\bSig$ is known up to a constant multiple, then \eqref{Eq.DT} is the LRT. The null is rejected in favor of the alternative at the level $\alpha $ if $T_{n}\geq c_{\alpha }$ the $\alpha$ level critical value. We say that the DT is consistent at $\bt\in\mathbb{R}^m$ if $\mathbb{P}_{\bt}\left( T_{n}\geq c_{\alpha }\right)\to 1$ as $n\to\infty$.

Understanding the geometry of the DT requires additional notation. First, for any cone $\mC$ let $\mC_{\bSig}^{\circ}$ denote its polar cone with respect to the inner product $\langle 
\boldsymbol{u},\boldsymbol{v}\rangle_{\bSig}=
\boldsymbol{u}^{T}\bSig^{-1}\boldsymbol{v}$, i.e., $\mathcal{C}_{\bSig}^{\circ}=\{\boldsymbol{u}\in
\mathbb{R}^{m}:\boldsymbol{u}^{T}\bSig^{-1}\boldsymbol{v}\leq 0,\forall \boldsymbol{v}\in \mC\}$. For convenience we shall write $\mathcal{C}^{\circ }$ instead of $\mC_{\bSig}^{\circ} $ whenever no ambiguity arises. Next, note that the continuity of projections onto convex sets and the continuous mapping theorem imply that $T_{n}=n\Delta +O_{p}(\sqrt{n})$  where
\begin{align} \label{Eq.Delta}
\Delta=\|\bt-\Pi_{\bSig}(\bt|\Theta_{0})\|_{\bSig}^{2}-\|
\bt-\Pi_{\bSig}(\bt|\Theta_{1})\|_{\bSig}^{2},   
\end{align}
so $T_{n}\to\infty$ if and only if $\Delta >0$ in which case $\mathbb{P}_{\bt}(T_{n}\geq c_{\alpha })\rightarrow 1$. Using \eqref{Eq.Delta} we have:
\begin{theorem} \label{Thm-Geom}
In Type A Problems, the DT is consistent provided 
\begin{equation} \label{Eq.Con.DT.TypeA}
\bt \notin (\mC\cap\mL^{\bot})^{\circ}.
\end{equation}
In Type B Problems the DT is consistent for all $\bt \notin\mC$.
\end{theorem}

\begin{remark}
Theorem \ref{Thm-Geom} has been informally stated but not proved Silvapulle and Sen (2005). 
\end{remark} 

An important special case of \eqref{Eq.Con.DT.TypeA} arises when the cone $\mC$ is defined by a finite set of linear inequalities in which case it is referred to as polyhedral cone, i.e., $\mL=\{\bt\in\mathbb{R}^{m}:\pmb{R\theta }=\bzero\} $ and $\mC=\{\bt\in\mathbb{R}^{m}:\pmb{R\theta } \geq \bzero\} $ for some $p\times m$ restriction matrix $\boldsymbol{R}$. Set $\pmb{\eta}=\pmb{R\theta}$ and rewrite \eqref{Eq.Hypotheses.TypeA} as $H_{0}:\pmb{\eta}\in\mathcal{M}$ versus $H_{1}:\pmb{\eta}\in\mathcal{Q}\backslash\mathcal{M}$ where $\mathcal{M}=\{\bzero\}$ and $\mathcal{Q}=\mathbb{R}_{+}^{p}\equiv\{\pmb{\eta}\in\mathbb{R} ^{p}:\pmb{\eta }\geq \bzero\}$. Let $\pmb{W}_{n}=\pmb{RS}_{n}$ be an estimator of $\pmb{\eta}$. It follows from \eqref{Eq.Sn} that $\sqrt{n}(\pmb{W}_{n}-\pmb{\eta })\Rightarrow \mathcal{N}_{p}(\bzero,\pmb{R\Sigma R}^{T})$. Next note that $\mathcal{M}^{\bot}=\mathbb{R}^{p}$ so $\mathcal{Q}\cap\mathcal{M}^{\bot}$ is nothing but $\mathcal{Q}$. By Proposition 3.12.8 in Silvapulle and Sen (2005) and a bit of algebra it can be shown that the polar cone of $\mathcal{Q}$ with respect to $\pmb{R\Sigma R}^{T}$ is given by $\mathcal{Q}^{\circ } =\{ \pmb{\eta}\in\mathbb{R}^{p}:\pmb{\eta}^{T}(\pmb{R\Sigma R}^{T})^{-1}\leq \bzero\}$. Substituting $\pmb{\eta}=\pmb{R\theta}$ we conclude that the DT is not consistent provided $\bt$ satisfies $\bt^{T}\pmb{R}^{T}(\pmb{R}\bSig\pmb{R}^{T})^{-1}\leq \bzero$, a condition that is easy to check.

\begin{example} \label{Ex.Positivity.I}
Consider testing 
\begin{equation} \label{Eq.TypeA.Canonical}
H_{0}:\bt = \bzero \text{ versus } H_{1}: \bt\in \mathbb{R}_{+}^{2}\setminus\bzero
\end{equation}
based on a sample $\pmb{X}_1,\pmb{X}_2,\ldots$ from $\mathcal{N}_{2}(\bzero,\pmb{I})$ distribution. Applying Theorem \ref{Thm-Geom} shows that the DT, which coincides with the LRT, is consistent whenever $\bt \in \mathbb{R}^2 \setminus \mathbb{R}_{-}^2$, where $\mathbb{R}_{-}^2$ denotes the negative quadrant. Therefore, if $\bt\in \{\bt:\theta _{1}>0,\theta _{2}\leq 0\}\cup \{\bt:\theta _{1}\leq 0,\theta _{2}>0\}$, i.e., $\bt$ belongs to the second or fourth quadrants, then the DT is consistent although $\bt$ does not belong to the alternative. For such $\bt$ and all $n\in\mathbb{N}$ there is a possibility of a Type III Error. 
\end{example}

Let $\mathcal{A}_{n}(\Theta_{0},\Theta_{1},\alpha)$ and $\mathcal{R}_{n}(\Theta_{0},\Theta_{1},\alpha)$  denote the acceptance and rejection regions, respectively, for testing $H_0: \bt\in \Theta_0$ versus $H_0: \bt\in \Theta_1\setminus\Theta_0$ using \eqref{Eq.DT} at the level $\alpha$. The following theorem describes the DT--based acceptance regions for Type A and B Problems. In what follows the symbol $\oplus$ denotes the Minkowski Sum, i.e., for any sets $\mathcal{U}$ and $\mathcal{V}$ define $\mathcal{U}\oplus \mathcal{V}=\{u+v:u\in \mathcal{U},v\in \mathcal{V}\}$, and $\mathrm{Ball}_{\bSig}(\bzero,c) =\{\pmb{x}:\|\pmb{x}\|_{\bSig}^{2}<c^{2}\}$ is the open ball with radius $c$ with respect to the norm $\|\boldsymbol{\cdot}\|_{\bSig}$.

\begin{theorem} \label{Thm-Acceptance.Regions}
For Type A Problems we have: 
\begin{equation} \label{Eq.1.Thm.Acc}
\mathcal{A}_{n}(\mL,\mC,\alpha)=(\mC\cap\mL^{\bot})_{\bSig_{n}}^{\circ}\oplus \mathrm{Ball}_{\bSig_{n}}(\bzero,\sqrt{\frac{c_{\alpha }}{n}}),
\end{equation}
whereas for Type B Problems we have: 
\begin{equation} \label{Eq.2.Thm.Acc}
\mathcal{A}_{n}(\mC,\mathbb{R}^{m},\alpha)=\mC\oplus\mathrm{Ball}_{\bSig_{n}}(\bzero,\sqrt{\frac{c_{\alpha }}{n}}).  
\end{equation}
\end{theorem}

Theorems \ref{Thm-Geom} and \ref{Thm-Acceptance.Regions} are closely related. In fact \eqref{Eq.1.Thm.Acc} can be obtained from \eqref{Eq.Con.DT.TypeA} by adding the Minkowski addition of small ball to \eqref{Eq.Con.DT.TypeA}. This procedure is sometimes referred to as $\varepsilon-$fattening. One difference between Theorems \ref{Thm-Geom} and \ref{Thm-Acceptance.Regions} is that the acceptance regions are computed with respect to the estimated variance $\bSig_{n}$ whereas consistency is calculated with respect to the true variance $\bSig$. The difference between the two is small as $\mathcal{A}_{n}(\mL,\mC,\alpha )\rightarrow (\mC\cap\mL^{\bot})^{\circ}$ as $n\rightarrow \infty $. 

\medskip

In certain cases, one can obtain more explicit and intuitive characterization than those provided by Theorems \ref{Thm-Geom} (and \ref{Thm-Acceptance.Regions}). For example, consider the simplest possible ANOVA model 
\begin{equation} \label{Eq.ANOVA}
Y_{ij}=\theta_{i}+\varepsilon_{ij},  
\end{equation}
where $i=1,\ldots,K,$ $j=1,\ldots ,n_{i}$ and $\varepsilon_{ij}$ are IID with mean $0$\ and variance $\sigma^{2}$. Typically, it is assumed that under the null all means are equal, i.e., $\bt\in\mL$ where $\mL=\{\bt\in\mathbb{R}^{K}:\theta_{1}=\cdots=\theta_{K}\}$. The most common ordered alternatives are the simple, tree and umbrella orders specified by the cones $\mC_{s}=\{\bt\in\mathbb{R}^{K}:\theta_{1}\leq\cdots\leq\theta _{K}\}$, $\mC_{t}=\{\bt\in \mathbb{R}^{K}:\theta_{1}\leq \theta_{2},\ldots,\theta_{1}\leq \theta_{K}\}$ and $\mC_{u}=\{\bt\in \mathbb{R}^{K}:\theta _{1}\leq \cdots \leq \theta _{p}\geq \cdots \geq \theta_{K}\}$, respectively. All three alternative hypotheses arise frequently in practice and were instrumental in motivating the development of ORI; see Barlow et al. (1972) and van Eeden (2006) for surveys of early work in the area.

\begin{theorem} \label{Thm-3}
The DT for $H_{0}:\bt\in\mL$ versus $H_{1}:\bt\in\mC_{s}\backslash \mL$ is consistent provided that for some $1\leq i \leq K-1$ we have
\begin{equation} \label{Eq.Thm3}
\max_{1\leq s\leq i}\mathrm{Av}(\bt,s,i)
<\min_{i+1\leq t\leq K}\mathrm{Av}(\bt,i+1,t)
\end{equation}
where $\mathrm{Av}(\bt,u,v)=\sum_{u}^{v}w_{j}\theta
_{j}/\sum_{u}^{v}w_{j}$ and $w_{j}=\lim n_{j}/n\in (0,1) $ where $n=\sum_{j=1}^{K}n_{j}$ and $j=1,\ldots ,K$.
\end{theorem}

Equation \eqref{Eq.Thm3} states that the DT is consistent if there exists an index $i$ that partitions the means into at least two level--sets. If so, the means are, in a weak sense, increasing on average. Equivalently, since $\Pi_{\bSig}(\bt|\mL) = \mathrm{Av}(\bt,1,K)\pmb{1}_{K}$, where $\pmb{1}_{K}=(1,\ldots,1)^{T}$, must differ from $\Pi_{\bSig}(\bt|\mC_{s})$ when the test is consistent, it follows that the DT has no power on  $(\mC_{s}\cap\mL^{\bot })^{\circ} = \{\,\bt\in\mathbb{R}^{m}:\Pi_{\bSig}(\bt |\mC_{s})\in \mathrm{span}\{\pmb{1}_{K}\}\,\}$.

\begin{example} \label{Ex.ANOVA.K=3}
Specifically, suppose that $K=3$ and $n_{1}=n_{2}=n_{3}$. It is easy to verify that Equation \eqref{Eq.Thm3} holds for $i=1$  provided $\theta_{1}<\min\{\theta_{2},(\theta_{2}+\theta_{3})/2\}$ and therefore, either: $(a)$ $\theta_{1}<\theta_{2}\leq \theta_{3}$; or $(b)$ $\theta_{1}<\theta_{2}>\theta_{3}$ and $\theta_{1}<(\theta_{2}+\theta_{3})/2$. Similarly, when $i=2$ Equation \eqref{Eq.Thm3} holds whenever $ \max \{(\theta _{1}+\theta _{2})/2,\theta _{2}\}<\theta _{3}$ and therefore either $(c)$ $\theta_{1}\leq\theta_{2}<\theta_{3};$ or $(d)$ $\theta_{1}>\theta_{2}<\theta_{3}$ and $(\theta_{1}+\theta_{2})/2<\theta_{3}$. Clearly, if $\bt$ satisfies $(a)$ or $(c)$ then $\bt\in\mC_{s}$. However if $\bt$, satisfies $(b)$ or $(d)$, then $\bt\notin\mC_{s}$ and a Type III error will occur with probability tending to unity as $n\to\infty$.  
\end{example}

With slight modifications, the proof of Theorem \ref{Thm-3} can be adapted to deal with the tree and umbrella order. In particular, the DT for the tree order is consistent provided $\theta_{1}<\max \{\theta_{2},\ldots ,\theta_{K}\}$, i.e., if $\theta_{1}<\theta_{i}$ for some $i\in \{2,\ldots,K\}$. If so, for large $n$ the constrained estimator of $\theta_{1}$ satisfies
\begin{equation*}
\theta_{1,n}^{\ast}=\frac{\sum_{i\in J}w_{j}\theta _{i}}{\sum_{i\in J}w_{j}}+o_{p}(1)
\end{equation*}
where $J=\{ 1\leq j\leq K:\theta_{j}\leq \theta_{1}\}$ whereas the constrained estimator of $\theta _{i}$ for $i\notin J$ satisfies $\theta_{i,n}^{\ast}=\theta _{i}+o_{p}(1)$. Consequently for any $i \notin J$ we have $\mathbb{P}(\theta_{1,n}^{\ast}<\theta_{i,n}^{\ast})\rightarrow 1$ and it follows from the arguments in the proof of Theorem \ref{Thm-3} that the DT is consistent. The situation for the umbrella order is a bit more complicated. However, it can be demonstrated that the DT is consistent if either: $\left(i\right) $ the up--branch, i.e., the subvector $\left(\theta _{1},\ldots ,\theta _{p}\right)$, satisfies the conditions of Theorem \ref{Thm-3}, i.e., for some $1\leq i<p$ we have $\max_{1\leq s\leq i}\mathrm{Av}(\bt,s,i) <\min_{i+1\leq t\leq p}\mathrm{Av}(\bt,i+1,t) ;$ or if $(ii)$ the down branch, i.e., the subvector $(\theta_{p},\ldots,\theta_{K})$ satisfies the reverse condition, i.e., for some $p\leq i\leq K-1$ we have $\min_{p\leq s\leq i}\mathrm{Av}(\bt,s,i) >\max_{i+1\leq t\leq K}\mathrm{Av}(\bt,i+1,t)$.

\medskip

Studying the geometry of the DT enables an explicit  characterization of the set on which the DT is consistent and, therefore, also the set on which Type III errors will occur with probability increasing to one as $n \to \infty$. Examples \ref{Ex.Positivity.I} and \ref{Ex.ANOVA.K=3} demonstrate how Type III errors arise in common scenarios, emphasizing that Type III errors are not unfortunate pathological accidents or bizarre special cases, but the norm. In other words, Type III errors are ubiquitous in Type A Problems. Although the possibility of Type III errors in ORI is well known among researchers in the field, we are unaware of any formal investigation thereof. Clearly, Type III errors do not arise in Type B Problems where $\Theta_{0}\cup \Theta_{1}= \mathbb{R}^{m}$.

Although Type III errors are much less familiar than Type I \& II errors, their effect on the validity of our inferences are as severe. For example, a naive application of the DT may lead to the erroneous conclusion that a given treatment improves all outcomes under study, when, in fact, there is an improvement in a single outcome accompanied by deterioration in the remaining $m-1$ outcomes. Such errors inevitably result in poor decision making. For obvious reasons, Type III errors are particularly prevalent and dangerous in high--dimensional settings. In some applications, Type III errors are referred to as directional errors (Lehmann and Romano, 2005). For further discussion and additional perspectives, see Kaiser (1960), Shaffer (1972, 1990, 2002), Finner (1999), Oleckno (2008), Guo et al. (2010), Salkind (2010), Mayo and Spanos (2011), Grandhi et al. (2016), and Lin and Peddada (2024). 

\section{Safe tests for Type A Problems} \label{Sec.SafeTest}

We start by defining safety in testing.

\begin{definition}
A $\alpha$ level test $T_{n}$ with rejection region $\mathcal{R}$ is said to be safe if for all $\alpha\in (0,1) $ and each fixed $\bt\notin \Theta_{1}$ we have 
\begin{equation*}
\lim_{n}\mathbb{P}_{\bt}(T_{n}\in \mathcal{R}) = 0.
\end{equation*}
\end{definition}

Thus, a test is safe if the probability that it commits Type III errors decreases to $0$ as $n\rightarrow \infty $.

\begin{theorem} \label{Thm-DT.not.SAFE}
The DT \eqref{Eq.DT} is not safe for testing \eqref{Eq.Hypotheses.TypeA} in Type A Problems.
\end{theorem}

Theorem \ref{Thm-DT.not.SAFE} shows that the DT is not safe. Hence, it validates the principled objections to the use of ORI--based methods, as discussed in the Introduction and further demonstrated in Examples \ref{Ex.Positivity.I}  and \ref{Ex.ANOVA.K=3}. In the following, a safe test that alleviates these concerns is introduced and studied.

\subsection{Formulation} \label{SubSec.Safe.General}

We introduce safe tests for Type A Problems; more  general safe tests are briefly discussed in Section \ref{Sec.Summary}. We start with some notation. Let $t_n$ denote the realized value of $T_n$, the DT for testing \eqref{Eq.Hypotheses.TypeA}. Denote the associated p--value by $\alpha ^{\ast}=\mathbb{P}_{\bzero}(T_{n}\geq t_{n})$. Next, consider an auxiliary system of hypotheses; in general these are of the form 
\begin{align*} 
H_{0}':\bt\in\Theta_0\cup\Theta_1 \text{  versus  } H_{1}':\bt\notin\Theta_0\cup\Theta_1.    
\end{align*}
Note that when $(\Theta_0,\Theta_1)=(\mL,\mC)$, i.e., in Type A problems, the auxilliary hypotheses reduce to \eqref{Eq.Hypotheses.TypeB}, a Type B testing problem. Let $T_n^{'}$ denote the corresponding DT and let $t_n^{'}$ be its realized value. It is well known that for any $\gamma \in (0,1)$ the critical value $c_{\gamma}'$ for testing \eqref{Eq.Hypotheses.TypeB} solves $\gamma =\sup_{\bt\in \mC}\mathbb{P}_{\bt}(T_{n}^{'}\geq c_{\gamma}') =\mathbb{P}_{\bzero}(T_{n}^{'}\geq c_{\gamma}')$; correspondingly the associated p--value is $\gamma^{\ast}=\mathbb{P}_{\bzero}(T_{n}^{'}\geq t_{n}^{'})$.

The pair $(\gamma^{\ast},\alpha^{\ast})$ 
reports on the outcome of the two tests and summarizes the evidence in the data. Consider the mapping $(\gamma^{\ast},\alpha^{\ast}) \mapsto (D_{1},D_{2})$
where $D_{1}=\mathbb{I}_{\{{\gamma^{\ast}\geq \gamma}\}}$ and $D_{2}=\mathbb{I}_{\{{\alpha{\ast}\leq \alpha}\}}$ for some specified values of $\gamma$ and $\alpha$. In particular $D_{1}=1$ indicates that the auxiliary null is not rejected at the level $\gamma$. Equivalently, a certificate of validity of level $\gamma$, is issued. The event $D_{2}=1$ indicates that the original null is rejected at the level $\alpha$. See Table \ref{Table.Decisions} for the relevant possibilities.

\FloatBarrier
\begin{table}[!ht]
\centering
\caption{The decision space}
\label{Table.Decisions}
\begin{tabular}{l|l|l}
Certificate $(D_{1})$ & Original Test $(D_{2})$ & 
Conclusion \\ \hline
\multicolumn{1}{c|}{$1$} & \multicolumn{1}{|c|}{$1$} & Safely, reject the Null. \\ 
\multicolumn{1}{c|}{$1$} & \multicolumn{1}{|c|}{$0$} & Do not reject the Null. \\ 
\multicolumn{1}{c|}{$0$} & \multicolumn{1}{|c|}{$1$} & A likely Type III error. Revisit assumptions. \\ 
\multicolumn{1}{c|}{$0$} & \multicolumn{1}{|c|}{$0$} & Do not reject the Null. Revisit assumptions.
\end{tabular}
\end{table}

Henceforth we shall refer to the tests $T_{n}$ and $T_{n}^{'}$ as the base--tests. In the literature, e.g., Raubertes et al. (1986), these tests are sometimes denoted by $T_{n,01}$ and $T_{n,12}$. We shall combine the base--tests to derive a composite safe test. For convenience we will denote the acceptance regions of the base tests by $\mathcal{A}_{n}=\mathcal{A}_{n}(\mL,\mC,\alpha)$ and $\mathcal{A}_{n}^{'}=\mathcal{A}_{n}(\mC,\mathbb{R}^{m},\gamma)$, respectively. The corresponding rejection regions are accordingly labeled. 

\begin{definition} \label{Def-Safe}
Fix $\alpha$ and $\gamma$. Let $T_{n}^{\mathrm{SAFE}}$ be a test for \eqref{Eq.Hypotheses.TypeA} with rejection region 
\begin{equation} \label{Eq.Rn.SAFE}
\mathcal{R}_{n}^{\mathrm{SAFE}}=\mathcal{R}_{n}^{\mathrm{SAFE}}(\alpha,\gamma) =\mathcal{A}_{n}^{'} \cap \mathcal{R}_{n}=\{\pmb{S}_{n}\in \mathbb{R}^{m}:T_{n}^{'}<c_{\gamma}',T_{n}\geq c_{\alpha}\}.
\end{equation}
\end{definition}
It immediately follows that $\mathcal{R}_{n}^{\mathrm{SAFE}}=\{\pmb{S}_{n}\in \mathbb{R}^{m}:T_{n}^{\mathrm{SAFE}}\geq c_{\alpha }\}$ where $T_{n}^{\mathrm{SAFE}}=T_{n}\mathbb{I}_{\{T_{n}^{'}<c_{\gamma }'\}}$. Hence, associate $T_{n}^{\mathrm{SAFE}}$ with a non--negative RV and reject the null only when $T_{n}^{\mathrm{SAFE}}\geq c_{\alpha }^{\mathrm{SAFE}}$ where $c_{\alpha }^{\mathrm{SAFE}}$ solves $\alpha =\mathbb{P}(T_{n}^{\mathrm{SAFE}}\geq c_{\alpha}^{\mathrm{SAFE}})$. 

\subsection{A safe test in two dimensions}

To fix ideas, consider Example \ref{Ex.Positivity.I}. It is well known that the DT for \eqref{Eq.TypeA.Canonical} is $T_{n}=n\|\Pi(\overline{\pmb{X}}_{n}|\mathbb{R}_{+}^{2})\| _{\bI}^{2}$. Moreover, under the null the DT is distributed as the mixture $\chi _{0}^{2}/4+\chi _{1}^{2}/2+\chi _{2}^{2}/4$ where $\chi _{i}^{2}$ is a chi--square RV with $i$ degrees of freedom. As noted in Example \ref{Ex.Positivity.I} if $\bt$ is in the second or fourth quadrants, then the DT is consistent although $\bt$ does not belong to the alternative. In other words, for every $n\in\mathbb{N}$ there is a nonzero probability of a Type III error, and as $n\to\infty$ a Type III error will occur with probability tending to one. The auxiliary hypotheses reduce to 
$H_{0}:\bt\in\mathbb{R}_{+}^{2}$ versus $H_{1}:\bt\notin \mathbb{R}_{+}^{2}$ tested using the DT $T_{n}^{'}=n\|\Pi(\overline{\pmb{X}}_{n}|\mathbb{R}_{-}^{2})\|_{\bI}^{2}$. It is well known that in this specific setting $T_n$ and $T_n'$ have the same distribution, hence $c_{\alpha} = c_{\alpha}'$ for all $\alpha\in(0,1)$.

By Theorem \ref{Thm-Acceptance.Regions} the acceptance regions of the base--tests are $\mathcal{A}_{n}=\mathbb{R}_{-}^{2}\oplus\mathrm{Ball}(\bzero,\sqrt{c_{\alpha }/n})$ and $
\mathcal{A}_{n}^{'}=\mathbb{R}_{+}^{2} \oplus \mathrm{Ball}(\bzero,\sqrt{c_{\gamma }/n})$, respectively. Let $\mathfrak{L}$ and $\mathfrak{L}^{'}$ denote the curves representing the boundaries of $\mathcal{A}_{n}$ and $\mathcal{A}_{n}^{'}$ respectively. Note that the curve $\mathfrak{L}$, in red, is at a distance $\sqrt{c_{\alpha }/n}$ from the negative orthant $\mathbb{R}_{-}^{2}$ whereas $\mathfrak{L}^{'}$, in blue, is at a distance $\sqrt{c_{\gamma}/n}$ from the positive orthant $\mathbb{R}_{+}^{2}$. Clearly $\mathcal{A}_{n}$ is lower set while $\mathcal{A}_{n}^{'}$ is an upper set (Shaked and Shanthikumar, 2007). Figure \ref{Fig.Regions} provides a graphical illustration. 

\FloatBarrier
\begin{figure}[ht!] 
\centering
\begin{tikzpicture}[domain=0:4]
    \draw[<->] (-5,0) -- (5,0) node[right] {$\theta_1$};
    \draw[<->] (0,-5) -- (0,5) node[above] {$\theta_2$};
    
    \draw[domain=0:5, smooth, variable=\x, blue,thick] plot ({\x}, {-2.5});
    \draw[domain=0:5, smooth, variable=\y, blue,thick] plot ({-2.5}, {\y});
    \draw[domain=-5:0, smooth, variable=\x, red,thick] plot ({\x}, {2.5});
    \draw[domain=-5:0, smooth, variable=\y, red,thick] plot ({2.5}, {\y});
    \draw [red,thick,domain=0:90] plot ({2.5*cos(\x)}, {2.5*sin(\x)});
    \draw [blue,thick,domain=180:270] plot ({2.5*cos(\x)}, {2.5*sin(\x)});
    
    \node at (-4, 4) {$R_1$};
    \node at (-2, 4) {$R_2$};
    \node at (3, 3) {$R_3$};
    \node at (4, -2) {$R_4$};
    \node at (4, -4) {$R_5$};
    \node at (2, -4) {$R_6$};
    \node at (-3, -3) {$R_7$};
    \node at (-4, 2) {$R_8$};
    \node at (-1, 1) {$R_9$};
    \node at (1, 1) {$R_{10}$};
    \node at (1, -1) {$R_{11}$};
    \node at (-1, -1) {$R_{12}$};
\end{tikzpicture}
\caption{The acceptance rejection region of $T_n$ lies below the (red) curve $\mathfrak{L}$ whereas the acceptance region of $T_n^{'}$ lies above the (blue) curve $\mathfrak{L}^{'}$. For small values of $\gamma$ the curve $\mathfrak{L}^{'}$ lies at large distance from $\mathbb{R}_{+}^2$ whereas when $\gamma$ is large the curve closely hugs $\mathbb{R}_{+}^2$.}
\label{Fig.Regions}
\end{figure}
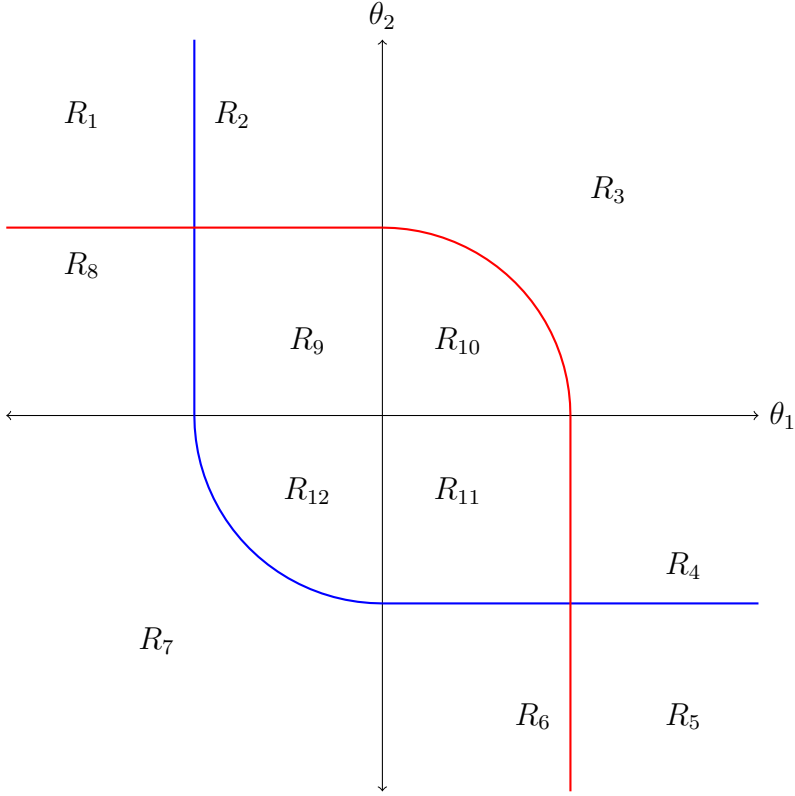
\FloatBarrier

\noindent The curves $\mathfrak{L}$ and $\mathfrak{L}^{'}$ together with the axes delineate twelve regions in $\mathbb{R}^{2}$ denoted $R_{1},\ldots ,R_{12}.$ The region $R_{1}$ is located in the upper left corner of Figure \ref{Fig.Regions} and all other regions are numbered according to their location on the imaginary spiral beginning in $R_{1}$ and terminating at $R_{12}.$ The rejection region of $T_n$ is $\mathcal{R}=R_{1}\cup R_{2}\cup R_{3}\cup R_{4}\cup R_{5}$ and
that of $T_n^{'}$ is $\mathcal{R}'=R_{5}\cup R_{6}\cup R_{7}\cup R_{8}\cup R_{1}$. Both tests accept on $R_{9}\cup R_{10}\cup R_{11}\cup R_{12}$. The regions $R_{1}$ and $R_{5}$ are in the rejection region of both tests; moreover, for all $\gamma\in(0,1)$ Type III errors occur on $R_{1}\cup R_{5}$. The rejection region of $T_{n}^{\mathrm{SAFE}}$, see Definition \ref{Def-Safe}, is $R_{2}\cup R_{3}\cup R_{4}$. Further observe that under the original null and for any fixed significance levels $\alpha $ and $\gamma $ the level of the composite test, denoted by $\alpha^{\mathrm{SAFE}}$, is
\begin{eqnarray*}
\alpha ^{\mathrm{SAFE}} &=&\mathbb{P}_{\bzero}\
(\overline{\pmb{X}}_{n}\in R_{2}\cup R_{3}\cup R_{4})=\mathbb{P}_{\bzero}(\overline{\pmb{X}}_{n}\in \mathcal{R})-\mathbb{P}_{\bzero}(\overline{\pmb{X}}_{n}\in R_{1}\cup R_{5}) \\
&=&\alpha -(\mathbb{P}_{\bzero}(\overline{\pmb{X}}_{n}\in R_{1})+\mathbb{P}_{\bzero}(\overline{\pmb{X}}_{n}\in R_{5})) \\
&=&\alpha-\mathbb{P}_{\bzero}(\overline{X}_{1,n}\leq -
\sqrt{\frac{c_{\gamma }}{n}},\overline{X}_{2,n}\geq \sqrt{\frac{c_{\alpha }}{n}})
-\mathbb{P}_{\bzero}(\overline{X}_{1,n}\geq \sqrt{\frac{
c_{\alpha }}{n}},\overline{X}_{2,n}\leq -\sqrt{\frac{c_{\gamma }}{n}}) \\
&=&\alpha -\mathbb{P}_{\bzero}
(\sqrt{n}\overline{X}_{1,n}\leq -\sqrt{c_{\gamma }})\mathbb{P}_{
\bzero}(\sqrt{n}\overline{X}_{2,n}\geq \sqrt{c_{\alpha}})-\mathbb{P}_{\bzero}(\sqrt{n}\overline{X}_{1,n}\geq \sqrt{c_{\alpha}})\mathbb{P}_{\bzero}\sqrt{n}
\overline{X}_{2,n}\leq -\sqrt{c_{\gamma }})
\\
&=&\alpha -2(1-\Phi (\sqrt{c_{\alpha }}))(1-\Phi (\sqrt{c_{\gamma }})),
\end{eqnarray*}
where $\Phi$ is the DF of a standard normal RV. Clearly $\alpha^{\mathrm{SAFE}}\leq\alpha$ for all $\gamma\in (0,1)$.  Moreover, for any fixed $\gamma$ we can choose $\alpha$ so
\begin{equation} \label{Eq.Find.Level.Safe}
\alpha -2(1-\Phi (\sqrt{c_{\alpha }}))(1-\Phi (\sqrt{c_{\gamma }}))=\alpha ^{\mathrm{SAFE}}  
\end{equation}
for any prechosen level $\alpha^{\mathrm{SAFE}}\in (0,1)$. It turns out that for standard significance levels, i.e., whenever $\alpha \leq 0.1$ and $\gamma \leq 0.1,$ $\alpha^{\mathrm{SAFE}}$ is very close to $\alpha$. For example when $\alpha =\gamma=0.1$ we find that $\alpha ^{\mathrm{SAFE}}=0.0999$ and when $\alpha =0.1$ and $\gamma =0.5$ then $\alpha ^{\mathrm{SAFE}}=0.0988$. For smaller $\alpha$, i.e., $\alpha =0.05$ the differences between $\alpha$ and $\alpha^{\mathrm{SAFE}}$ are even smaller. 

A bit of reflection shows that the rejection region of $T_{n}^{\mathrm{SAFE}}$ can be reexpressed as $\{\overline{\pmb{X}}_{n}\in \mathbb{R}^{2}:T_{n}^{\mathrm{SAFE}}\geq c_{\alpha}^{\mathrm{SAFE}}\}$ where $c_{\alpha }^{\mathrm{SAFE}}$ is the $\alpha$--level critical value of the test $T_{n}^{\mathrm{SAFE}}$ which is given by
\begin{equation*}
T_{n}^{\mathrm{SAFE}}=n\{\overline{X}_{1,n}^{2}\mathbb{I}_{\{\overline{X}_{1,n}\ge 0,-\sqrt{\frac{c_{\gamma }}{n}} \le\overline{X}_{2,n} < 0\}} +\overline{X}_{2,n}^{2}\mathbb{I}_{\{-\sqrt{\frac{c_{\gamma }}{n}} \le \overline{X}_{1,n} <0, \overline{X}_{2,n} >0\}} +(\overline{X}_{1,n}^{2}+\overline{X}_{2,n}^{2})\mathbb{I}_{\{\overline{X}_{1,n}\geq 0,\overline{X}_{2,n}\geq
0\}}\}.
\end{equation*}
The value of $c_{\alpha }^{\mathrm{SAFE}}$ is obtained by solving the equation $\widetilde{\alpha}=\mathbb{P}(\chi_{1}^{2}\geq c)/2+\mathbb{P}(\chi _{2}^{2}\geq c)/4$ for $c$ where $\widetilde{\alpha }$ solves (\ref{Eq.Find.Level.Safe}) for a specified value of $\alpha^{\mathrm{SAFE}}$ guaranteeing that
$$ 
\mathbb{P}_{\bzero}(T_{n}^{\mathrm{SAFE}}\geq c_{\alpha }^{\mathrm{SAFE}}) = \alpha.  
$$
Moreover, since $\alpha^{\mathrm{SAFE}}\leq \alpha$, we have $c_{\alpha}^{\mathrm{SAFE}}\leq c_{\alpha}$.

Given an explicit form of $T_{n}^{\mathrm{SAFE}}$, see the display above, its distribution can be studied for any $\bt\in \mathbb{R}^{2}$. For example, it is immediate that for any $\bt\in \mC$ we have 
\begin{equation*}
\mathbb{P}_{\bt}(T_{n}^{\mathrm{SAFE}}\neq T_{n})\rightarrow 0
\end{equation*}
for any fixed $n$ as $\gamma \rightarrow 0$ and similarly for any fixed $\gamma $ as $n\rightarrow \infty$. Hence $T_{n}$ and $T_{n}^{\mathrm{SAFE}}$ are expected to perform similarly on the alternative. However, for any $\bt$ in the second or fourth quadrant it is easy to verify that $T_{n}^{\mathrm{SAFE}}=0$ for all $\gamma \in \left(0,1\right)$ if $n$ is large enough. Therefore, $T_{n}^{\mathrm{SAFE}}$, above, is indeed a safe test. 

\subsection{The general case}

We conclude this section with a general Theorem addressing safe tests for Type A Problems. 

\begin{theorem} \label{Thm-4}
Consider testing \eqref{Eq.Hypotheses.TypeA} using $T_{n}^{\mathrm{SAFE}}$ for some fixed levels $\alpha $ and $\gamma$. Then
\begin{equation} \label{Eq.1.Thm.4}
\alpha^{\mathrm{SAFE}}=\sup_{\bt\in\mL}\mathbb{P}_{\bt}(T_{n}^{\mathrm{SAFE}}\geq c_{\alpha })\leq \alpha .
\end{equation}
Moreover, the test $T_{n}^{\mathrm{SAFE}}$ is safe. Furthermore, for all $\bt\in \mathrm{int}(\mC\backslash\mL)$ we have
\begin{equation} \label{Eq.2.Thm.4}
\mathbb{P}_{\bt}(T_{n}\geq c_{\alpha })\leq \mathbb{P}_{\bt}(T_{n}^{\mathrm{SAFE}}\geq c_{\alpha }^{\mathrm{SAFE}})+o(1)  
\end{equation}
where $c_{\alpha }^{\mathrm{SAFE}}$ is the $\alpha$--level critical value of $T_{n}^{\mathrm{SAFE}}$.
\end{theorem}

Theorem \ref{Thm-4} shows that the test $T_{n}^{\mathrm{SAFE}}$ with rejection region \eqref{Eq.Rn.SAFE} is a safe test. Moreover, for large samples the test based on $T_{n}^{\mathrm{SAFE}}$ is more powerful than the test based on $T_{n}$ in the interior of $\mathcal{C}$. However, in finite samples the safe test may incur some power loss near the boundary of $\mC$. This loss typically depends on $\gamma$ and the sample size $n$. It is also clear that for finite samples the possibility of Type III errors while using $T_{n}^{\mathrm{SAFE}}$ persists, but to a lesser degree and in a smaller subset of the parameter space compared to $T_{n}$.

Next, suppose that $\mC$ is a polyhedral cone of the form $\mC=\{\bt\in\mathbb{R}^{m}:\pmb{R\theta } \geq \bzero\} $ for some $p\times m$ restriction matrix $\boldsymbol{R}$. It follows that testing \eqref{Eq.Hypotheses.TypeA} using $\pmb{S}_n$ is equivalent to testing $H_{0}:\pmb{\eta} = \bzero$ versus $H_{1}:\pmb{\eta}\in\mathbb{R}_{+}^{p}\setminus \{\bzero\}$ using $\pmb{W}_n$ where $\sqrt{n}(\pmb{W}_n-\pmb{\eta}) \Rightarrow \mathcal{N}(\bzero,\pmb{\Psi})$, $\pmb{\eta}=\pmb{R\theta}$ and $\pmb{\Psi}=\pmb{R} \pmb{\Sigma}\pmb{R}^{T}$. The corresponding auxiliary hypotheses are $H_{0}':\pmb{\eta} \in \mathbb{R}_{+}^{p}$ versus $H_{1}':\pmb{\eta}\notin\mathbb{R}_{+}^{p}$. The DTs for these systems are 
\begin{eqnarray*}
T_{n}&=& \|\Pi_{\pmb{\Psi}}(\pmb{W}|\mathbb{R}_{+}^{p})\|_{\pmb{\Psi}}^{2}+o_{p}(1),
\\
T_{n}'&=& \|\Pi_{\pmb{\Psi}}(\pmb{W}|(\mathbb{R}_{+}^{p})_{\pmb{\Psi}}^{\circ})\|_{\pmb{\Psi}}^{2}+o_{p}(1),
\end{eqnarray*}
where $\pmb{W}$ is a $\mathcal{N}(\bzero,\pmb{\Psi})$ RV. Now, by Lemma 3.13.6 in Silvapulle and Sen (2005), see also Raubertas et al. (1986), we can deduce that for any $c_1,c_2\ge 0$ we have
\begin{align} \label{Eq.Joint}
\mathbb{P}_{\bzero}(T_n \ge c_1,T_n' < c_2) =\sum_{j=0}^{p} w_{j}(p,\pmb{\Psi},\mathbb{R}_{+}^{p})\mathbb{P}(\chi_{j} \ge c_1)  \mathbb{P}(\chi_{p-j} < c_2) + o_p(1),      
\end{align}
where $w_j=w_{j}(p,\pmb{\Psi},\mathbb{R}_{+}^{p})$, $j=0,\ldots,p$ are nonnegative weights that sum to unity. 
Furthermore, \eqref{Eq.Joint} implies that $T_n^{\rm SAFE} \Rightarrow T^{\rm SAFE}$ where for any $t^{\rm SAFE}\ge0$ and fixed $c_{\gamma}'$ the tail of $T^{\rm SAFE}$ is given by $\mathbb{P}_{\bzero}(T^{\rm SAFE} \ge t^{\rm SAFE}) = \mathbb{P}_{\bzero}(T \ge t^{\rm SAFE},T' < c_{\gamma}')$ where $T$ and $T'$ are respectively, the distributional limits of $T_n$ and $T_n'$. 

Evaluating p--values, finding critical values and related inferential tasks accurately approximating the unknown quantities in \eqref{Eq.Joint}. It is well known, see Lemma 3.13.7 in Silvapulle and Sen (2005), that $w_j$ where $j=0,\ldots,p$ is the probability that the RV $\pmb{W}$ is projected onto a face of dimension $j$ of the cone $\mathbb{R}_{+}^{p}$. It follows that by generating a large sample $\hat{\pmb{W}}_1,\ldots,\hat{\pmb{W}}_N$ from $\mathcal{N}(\bzero,\pmb{\Psi}_n)$, where $\pmb{\Psi}_n=\pmb{R\Sigma_nR}^T$, we can estimate $w_j$ by $\hat{w}_j=N^{-1}\sum_{k=1}^{N}\mathbb{I}_{\{\Pi(\hat{\pmb{W}}_k|\mathbb{R}_{+}^{p})\in \mathcal{F}_j\}}$, i.e., the proportion of times the projection is in $\mathcal{F}_j$, the collection of faces of dimension $j$. Clearly, $\hat{w}_j \xrightarrow{p} w_j$ for all $j$ as $n \to \infty$ and $N \to \infty$. Given the simulated weights $\hat{w}_0,\ldots,\hat{w}_p$ we can find the rejection region of $T_n^{\rm SAFE}$. First, fix $\gamma$ and set $c_1=0$ in \eqref{Eq.Joint}. Solve the resulting equation, i.e., $1-\gamma =\sum_{j=0}^{p} \hat{w}_{j}\mathbb{P}(\chi_{p-j}\le c_{\gamma}')$ and denote the solution by $\hat{c}_{\gamma}'$. Next, for any fixed $\alpha$ solve the equation $\alpha = \sum_{j=0}^{p} \widehat{w}_{j}\mathbb{P}(\chi_{j} > c_{\alpha}^{\rm SAFE}) \mathbb{P}(\chi_{p-j}\le \hat{c}_{\gamma}')$ and denote the solution by $\hat{c}_{\alpha}^{\rm SAFE}$. Note that both equations can be readily solved by the bisection method. Plugging $\hat{c}_{\alpha}^{\rm SAFE}$ and $\hat{c}_{\gamma}'$ into \eqref{Eq.Rn.SAFE} we obtain an approximate $\alpha$--level rejection region for $T_n^{\rm SAFE}$. We have just shown that given $\gamma$ it is always possible to adjust the level of $T_n$ so $T_{n}^{\mathrm{SAFE}}$ has any prechosen level in $(0,1)$. Using a similar procedure it is always possible to adjust $\gamma$ so $T_{n}^{\mathrm{SAFE}}$ has level $\alpha^{\mathrm{SAFE}}\in $ $(0,\alpha)$. 

\section{Numerical results}\label{Sec.Numerics}

In this section the performance of the safe test is compared with that of DT in several experimental settings. In addition we provide an analysis of two well known examples from the literature.

\subsection{Simulations}

We conducted a simulation study to evaluate the performance of the proposed safe test. Specifically, we test \eqref{Eq.TypeA.Canonical} by generating samples of size $n$ from $\mathcal{N}_{2}(\bt,\bI)$. Seven possible experimental settings for the value of $\bt$, listed in Table \ref{Table.theta.values}, were considered. These include a null value $\bt_{0} = (0,0)^{T}$ along with six non--null values, all located on a circle with radius $\|\bt_{i}\|^{2}=3/4$ for $i = 1,\ldots, 6$. For each $i$, we report $\measuredangle(\pmb{e}_{1},\bt_{i})$, the angle between $\bt_i$ and $\pmb{e}_{1} = (1,0)$ the positive horizontal axis.
Observe that $\bt_{i}\in\mC$ for $i\in\{1,2,3\}$ and $\bt_{i}\notin\mathcal{C}$ for $i\in\{4,5,6\}$. 

\FloatBarrier
\begin{table}[!ht]
\centering
\caption{Experimental settings for the mean value $\bt$ used in the simulation study.}
\label{Table.theta.values}
\begin{tabular}{l| l|l|l|l|l|l |l}
Mean & $\bt_{0}$ & $\bt_{1}$ & $\bt_{2}$ & $\bt_{3}$ & $\bt_{4}$ & $\bt_{5}$ & $\bt_{6}$ \\ 
\hline
\multicolumn{1}{l|}{Angle} & $-$ & $45^{\circ }$ & $15^{\circ }$ & $0$ & $
-15^{\circ }$ & $-45^{\circ }$ & \multicolumn{1}{|l}{$-60^{\circ }$} \\ 
\multicolumn{1}{l|}{Region} & $\mathcal{L}$ & $\mathcal{C}$ & $\mathcal{C}$
& $\mathcal{C}$ & $\mathcal{C}^c$ & $\mathcal{C}^c$ & \multicolumn{1}{|l}{$%
\mathcal{C}^c$} \\ 
\end{tabular}
\end{table}
\FloatBarrier

Table \ref{Table.Sim.Results} reports on the power of $T_{n}$ and $T_{n}^{\mathrm{SAFE}}$ in the aforementioned experimental settings assuming $\alpha =\alpha ^{\mathrm{SAFE}}=0.05$, $n\in \{ 10,20,50\}$ and $\gamma\in\{0.1,0.05,0.01\}$. The test $T_{n}$ is applied as usual while in order to use $T_{n}^{\mathrm{SAFE}}$ we first find $c_{\alpha }^{\mathrm{SAFE}}$ as explained earlier. The power is calculated using $10^{6}$ simulation runs and rounded to the third significant digit. The first block of results are obtained at the null $\bt_{0}=(0,0)^{T}$; these show that the actual level of $T_{n}^{\mathrm{SAFE}}$ agrees with its nominal level for all values of $\gamma$. At $\bt_{1}$, located in the interior of $\mathbb{R}_{+}^{2}$, there is no difference between the power of $T_{n}$ and $T_{n}^{\mathrm{SAFE}}$ whereas at $\bt_{2}$ the powers of $T_{n}$ and $T_{n}^{\mathrm{SAFE}}$ are very similar. At $\bt_{3}$, which lies on the boundary of $\mathbb{R}_{+}^{2}$, the power of $T_{n}$ is slightly higher than the power of $T_{n}^{\mathrm{SAFE}}$; this difference is insignificant when $\gamma$ is small. When $i\in\{4,5,6\}$ we have $\bt_{i}\notin\mathcal{C}$ and the DT commits a Type III error with high probability. In fact, as $n$ grows, so does the likelihood of a Type III error. On the other hand, the probability of a Type III error drops precipitously when $T_{n}^{\mathrm{SAFE}}$ is used especially for large $n$ and values of $\bt$ which are far from the alternative.

\FloatBarrier
\begin{table}[!ht] \centering
\caption{Comparing the power of the DT and Safe test}
\label{Table.Sim.Results}
\begin{tabular}{ll|ll|ll|ll}
&  & \multicolumn{2}{|c|}{$n=10$} & \multicolumn{2}{|c|}{$n=20$} & \multicolumn{2}{|c}{$n=50$} \\  
Mean & \multicolumn{1}{|l|}{$\gamma $} & $T_{n}$ & $T_{n}^{\mathrm{SAFE}}$ & 
$T_{n}$ & $T_{n}^{\mathrm{SAFE}}$ & $T_{n}$ & $T_{n}^{\mathrm{SAFE}}$ \\ 
\hline
& \multicolumn{1}{|l|}{$0.1$} & $0.050$ & $0.048$ & $0.050$ & $0.049$ & $0.050$ & $0.048$ \\ 
$\bt_{0}$ & \multicolumn{1}{|l|}{$0.05$} & $0.050$ & $0.049$ & $0.050$ & $0.049$ & $0.050$ & $0.049$ \\ 
& \multicolumn{1}{|l|}{$0.01$} & $0.050$ & $0.050$ & $0.050$ & $0.050$ & $0.050$ & $0.050$ \\ \cline{1-8}
& \multicolumn{1}{|l|}{$0.1$} & $0.705$ & $0.705$ & $0.931$ & $0.931$ & $1.000$ & $1.000$ \\ 
$\bt_{1}$ & \multicolumn{1}{|l|}{$0.05$} & $0.706$ & $0.706$ & $0.932$ & $0.932$ & $1.000$ & $1.000$ \\ 
& \multicolumn{1}{|l|}{$0.01$} & $0.706$ & $0.706$ & $0.932$ & $0.932$ & $1.000$ & $1.000$ \\ \cline{1-8}
& \multicolumn{1}{|l|}{$0.1$} & $0.693$ & $0.687$ & $0.928$ & $0.924$ & $1.000$ & $0.999$ \\ 
$\bt_{2}$ & \multicolumn{1}{|l|}{$0.05$} & $0.693$ & $0.691$ & $0.928$ & $0.927$ & $1.000$ & $0.999$ \\ 
& \multicolumn{1}{|l|}{$0.01$} & $0.693$ & $0.693$ & $0.928$ & $0.928$ & $1.000$ & $1.000$ \\ \cline{1-8}
& \multicolumn{1}{|l|}{$0.1$} & $0.666$ & $0.640$ & $0.917$ & $0.879$ & $1.000$ & $0.957$ \\ 
$\bt_{3}$ & \multicolumn{1}{|l|}{$0.05$} & $0.665$ & $0.652$ & $0.917$ & $0.899$ & $1.000$ & $0.980$ \\ 
& \multicolumn{1}{|l|}{$0.01$} & $0.666$ & $0.664$ & $0.917$ & $0.914$ & $1.000$ & $0.996$ \\ \cline{1-8}
& \multicolumn{1}{|l|}{$0.1$} & $0.608$ & $0.528$ & $0.885$ & $0.711$ & $0.999$ & $0.635$ \\ 
$\bt_{4}$ & \multicolumn{1}{|l|}{$0.05$} & $0.609$ & $0.565$ & $0.885$ & $0.782$ & $0.999$ & $0.752$ \\ 
& \multicolumn{1}{|l|}{$0.01$} & $0.609$ & $0.598$ & $0.885$ & $0.856$ & $0.999$ & $0.907$ \\ \cline{1-8}
& \multicolumn{1}{|l|}{$0.1$} & $0.353$ & $0.182$ & $0.624$ & $0.160$ & $0.955$ & $0.020$ \\ 
$\bt_{5}$ & \multicolumn{1}{|l|}{$0.05$} & $0.354$ & $0.230$ & $0.623$ & $0.234$ & $0.955$ & $0.043$ \\ 
& \multicolumn{1}{|l|}{$0.01$} & $0.354$ & $0.299$ & $0.624$ & $0.392$ & $0.955$ & $0.140$ \\ \cline{1-8}
& \multicolumn{1}{|l|}{$0.1$} & $0.193$ & $0.071$ & $0.352$ & $0.041$ & $0.724$ & $0.002$ \\ 
$\bt_{6}$ & \multicolumn{1}{|l|}{$0.05$} & $0.192$ & $0.096$ & $0.352$ & $0.070$ & $0.724$ & $0.004$ \\ 
& \multicolumn{1}{|l|}{$0.01$} & $0.192$ & $0.143$ & $0.352$ & $0.147$ & $0.724$ & $0.021$
\end{tabular}
\end{table}
\FloatBarrier

These results corroborate our theoretical findings, demonstrating that the safe test is effective in limiting the occurrence of Type III errors even with small sample sizes, while simultaneously maintaining power comparable to the standard DT, even in when the true value of the parameter is close to the boundary.

\subsection{Illustrative examples\label{Sec.IllustrativeEx}}

In this section two well--known examples are examined from the perspective of safe testing.

\subsubsection{Testing for positivity}

We start by demonstrating that the region in the parameter space on which Type III errors occur depends on the variance matrix $\bSig$. Recall that the well known Hotelling $T^{2}$ test (Bilodeau and Brenner, 1999) rejects the null $H_{0}:\bt=\bzero$ when $n\pmb{S}_{n}^{T}\bSig_{n}^{-1}\pmb{S}_{n}$ is large, i.e., whenever $\|\pmb{S}_{n}\|_{\bSig_{n}}^{2}$ is large. The latter norm also plays a role in the DT \eqref{Eq.DT}. We demonstrate the subtle and surprising effect of $\bSig$ (or $\bSig_n$) on the DT. Recall that the DT for testing \eqref{Eq.TypeA.Canonical} when $m=2$ is consistent provided 
\begin{equation*}
\bt\notin(\mathbb{R}_{+}^{2})_{\bSig}^{\circ}=\{\bt\in \mathbb{R}^{2}:\bt^{T}\bSig^{-1}\pmb{v}\leq 0,\forall \pmb{v}\in \mathbb{R}_{+}^{2}\text{ }\}=\{\bt\in\mathbb{R}^{2}:\bt^{T}\bSig^{-1}\pmb{v}\leq 0, \pmb{v}\in \{\pmb{e}_{1},\pmb{e}_{2}\}\}
\end{equation*}
where $\pmb{e}_{1}=(1,0)^{T},\pmb{e}_{2}=(0,1)^{T}$. Suppose that $\bSig$ has an interclass correlations structure of  the form
\begin{equation} \label{Eq.Var.Sill}
\left( 
\begin{array}{cc}
1 & \rho \\ 
\rho & 1
\end{array}
\right)  
\end{equation}
for some $|\rho|<1$. It is easily verified that $
(\mathbb{R}_{+}^{2})_{\bSig}^{\circ}=\{\bt\in\mathbb{R}^{2}:\theta_{1}-\rho\theta_{2}\leq 0,\ \theta_{2}-\rho\theta_{1}\leq 0\}$. Figure \ref{Fig.PolarCones} plots $(\mathbb{R}_{+}^{2})_{\bSig}^{\circ}$ for $\rho \in \{+1/4,-1/4\}$.

\FloatBarrier
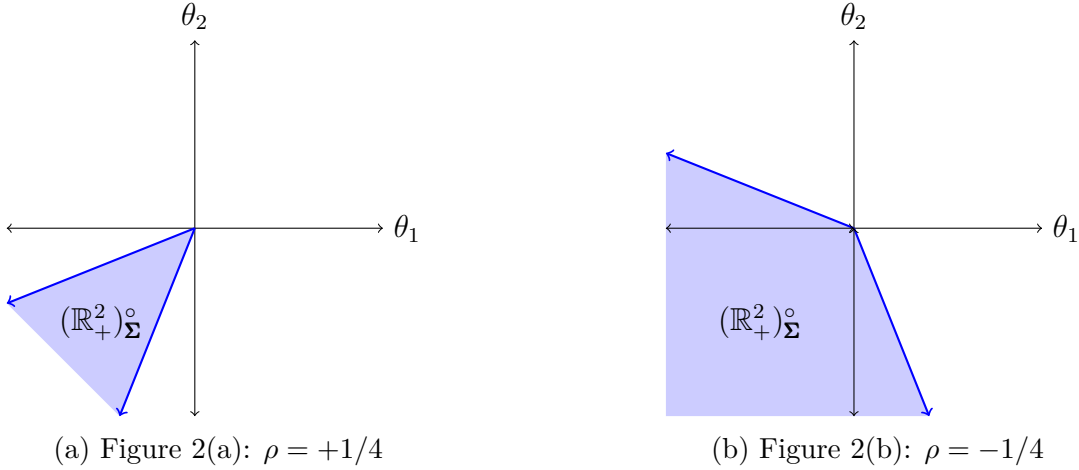
\begin{figure*}[htbp!]
\centering

\begin{subfigure}[b]{0.5\textwidth}
\centering
\begin{tikzpicture}
\draw[<->] (-2.5,0) -- (2.5,0) node[right] {$\theta_1$};
\draw[<->] (0,-2.5) -- (0,2.5) node[above] {$\theta_2$};
\fill[blue!20] (-2.5,-1) -- (0,0) -- (-1,-2.5) -- cycle;
\draw[<-,blue,thick, smooth] (-2.5,-1) -- (0,0);
\draw[<-,blue,thick, smooth] (-1,-2.5) -- (0,0);
\node at (-1.25, -1.25) {$(\mathbb{R}_{+}^{2})_{\bSig}^{\circ}$};
\end{tikzpicture}
\caption{Figure 2(a): $\rho = +1/4$}
\end{subfigure}%
\hfill
\begin{subfigure}[b]{0.5\textwidth}
\centering
\begin{tikzpicture}
\draw[<->] (-2.5,0) -- (2.5,0) node[right] {$\theta_1$};
\draw[<->] (0,-2.5) -- (0,2.5) node[above] {$\theta_2$};
\fill[blue!20] (-2.5,1) -- (0,0) -- (1,-2.5) -- (-2.5,-2.5) -- cycle;
\draw[<-,blue,thick, smooth] (-2.5,1) -- (0,0);
\draw[<-,blue,thick, smooth] (1,-2.5) -- (0,0);
\draw[<->] (-2.5,0) -- (0,0);
\draw[<->] (0,-2.5) -- (0,0);
\node at (-1.25, -1.25) {$(\mathbb{R}_{+}^{2})_{\bSig}^{\circ}$};
\end{tikzpicture}
\caption{Figure 2(b): $\rho = -1/4$}
\end{subfigure}

\caption{Plots of the cones $(\mathbb{R}_{+}^{2})_{\bSig}^{\circ}$, i.e., the regions in which the DT is not consistent, for $\rho\in\{+1/4,-1/4\}$. The cones, in blue, are bounded by their extreme rays which extend indefinitely.}  
\label{Fig.PolarCones}
\end{figure*}
\FloatBarrier

It is clear that if $\rho\geq 0$ then $(\mathbb{R}_{+}^{2})_{\bSig}^{\circ}\subseteq \mathbb{R}_{-}^{2}$ and  $(\mathbb{R}_{+}^{2})_{\bSig}^{\circ}\supseteq \mathbb{R}_{-}^{2}$ when $\rho\leq0$; if $\rho =0$ then $(\mathbb{R}_{+}^{2})_{\bSig}^{\circ}=\mathbb{R}_{-}^{2}$. Thus, the rejection region depends on the value of the correlation coefficient. Further observe that a Type III error may occur whenever $\bt\in(Q_2\cap Q_3\cap Q_4)\backslash (\mathbb{R}_{+}^{2})_{\bSig}^{\circ}$, where $Q_i$ is the $i^{th}$ quadrant of $\mathbb{R}^2$. It follows that the possibility of Type III errors increases when $\rho = +1/4$ compared to when $\rho = -1/4$, as the dependence structure alters the geometry of the rejection region. It is also clear from Figure \ref{Fig.PolarCones} that when $\rho =+1/4$, it is possible to reject the null even when both components of $\overline{\pmb{X}}_{n}$ are negative; this is not possible when $\rho =-1/4$. Thus, the correlation structure strongly influences the consistency of tests, and therefore also the possibility of Type III errors. Moreover, the power function, as a function of $\rho$, exhibits sharp phase transitions from regions where it is consistent to regions in which the test does not have any power; a phenomenon, not observed with unconstrained tests and which has not been fully recognized to date. It is also clear that the essence of this example holds in higher dimensions and more complexly structured variance matrices.

The preceding analysis is the key to understanding Silvapulle (1997) paper entitled \emph{A curious example involving the LRT for one sided hypotheses} which sparked a lively debate, cf. Perlman and Wu (1999). In that paper the hypotheses \eqref{Eq.TypeA.Canonical} was tested assuming $\pmb{X}_1, \ldots, \pmb{X}_5$ are IID $\mathcal{N}_2(\bt, \bSig)$, where $\overline{\pmb{X}}_{5} = (-3,-2)^{T}$ and $\bSig$ is of the form \eqref{Eq.Var.Sill} with $\rho = 0.9$. He observed that the LRT rejects the null hypothesis, whereas the individual one--sided tests, i.e., $H_{0}^{(i)}: \theta_{i} = 0$ versus $H_{1}^{(i)}: \theta_{i} > 0$ for $i \in {1,2}$, and therefore the Intersection--Union Test (Casella and Berger, 2024), do not. Silvapulle argued that these two diametrically opposed conclusions, reached using different testing procedures, may appear counterintuitive, but are not logically inconsistent. Nevertheless, he emphasized that this example serves as a caution against the indiscriminate application of one--sided tests.

We reanalyze this example from the perspective of safe testing. In fact, in this example $T_{n}=12.89$ and $c_{\alpha }$ solves $\alpha =\mathbb{P}(\chi _{1}^{2}\geq c_{\alpha })/2+(1-\cos^{-1}(\rho)/\pi )\mathbb{P}(\chi_{2}^{2}\geq c_{\alpha})$ which for $\alpha =0.05$ equals $4.915$. It follows that $T_{n}>c_{\alpha}$ and the associated p--value $\alpha^{\ast}$ is smaller than $0.001$. It is also easy to check that $T_{n}'>c_{\gamma }'$ for $\gamma \in \{0.1,0.05,0.01\}$. In fact, the p--value associated with the auxiliary hypotheses is highly significant, i.e., $\gamma^{\ast }<10^{-6}$. Thus, using the safe test, $T_{n}^{\mathrm{SAFE}}$, would not lead to the rejection of the null hypothesis but rather to a reevaluation of the original assumptions on the possible values of the parameter $\boldsymbol{\theta}$. Further note that since $\rho=9/10>0$ it is possible, as stated earlier, cf. Figure \ref{Fig.PolarCones}, to reject the null using $T_n^{\rm SAFE}$ even when both entries of $\overline{\pmb{X}}_{5}$ are negative. However, the distance of $\overline{\pmb{X}}_{5}$ from the null also plays an important role. In particular, $\overline{\pmb{X}}_{5}$ is sufficiently separated from the null, i.e., $(-3,-2)^T \notin \mathcal{A}_{5}(\bzero,\mathbb{R}_{+}^{2},\gamma)$ for all $\gamma \ge 10^{-5}$. We conclude that the inconsistency pointed out by Silvapulle's paper is easily understood and resolved by using the framework of safe testing. 

\subsubsection{Testing for stochastic order}

Cohen and Sackrowitz (1998) presented two tables comparing trinomial distributions with fixed marginals. Their data, appearing in Tables 5 and 6 in their paper, is displayed in  Table \ref{Table.CS.Data} below.

\FloatBarrier
\begin{table}[!ht] 
\centering
\caption{Tables 5 and 6 of Cohen and Sackrowitz (1998). Both sub--tables compare two trinomial distributions, one for the control group the other for the treatment group.}
\label{Table.CS.Data}
\begin{tabular}{llcccc}
&  & Worse & Same & Better & Total \\ \cline{3-6}
& Control & $5$ & $11$ & $1$ & $17$ \\ 
\underline{Table 5 of C \& S:} & Treatment & $3$ & $8$ & $4$ & $15$ \\ 
& Total & $8$ & $19$ & $5$ &  \\ 
\\ \hline
&  & Worse & Same & Better & Total \\ \cline{3-6}
& Control & $0$ & $16$ & $1$ & $17$ \\ 
\underline{Table 6 of C \& S:} & Treatment & $8$ & $3$ & $4$ & $15$ \\ 
& Total & $8$ & $19$ & $5$ & 
\end{tabular}
\medskip
\end{table}
\FloatBarrier

\noindent The objective was to test whether the outcomes distributions are ordered by treatment. Let $P=\left( p_{1},p_{2},p_{3}\right)$ and $Q=\left(q_{1},q_{2},q_{3}\right)$ denote the trinomial distribution of the control and treatment groups respectively. The stochastic order $P\preceq _{st}Q$, see Shaked and Shanthikumar (2007), holds provided $p_{1}\geq q_{1}$ and $p_{1}+p_{2}\geq q_{1}+q_{2}$. It follows that testing
\begin{equation} \label{Eq.Hypotheses.PQ}
H_{0}:P=_{st}Q ~~\text{ \ versus }~~ H_{1}:P\prec _{st}Q  
\end{equation}
is equivalent to testing $H_{0}:\pmb{\eta}=\bzero$ versus $H_{1}:\pmb{\eta}\in \mathbb{R}_{+}^{2}\backslash\{\bzero\}$  where $\pmb{\eta}=\pmb{R\theta }$ with
\begin{equation*}
\bt=\left(p_{1},p_{1}+p_{2},q_{1},q_{1}+q_{2}\right)^{T}\text{ and } \boldsymbol{R}=
\left( 
\begin{array}{cccc}
1 & 0 & -1 & 0 \\ 
0 & 1 & 0 & -1
\end{array}
\right). 
\end{equation*}
Let $\pmb{S}_{n}$ denote the MLE of $\bt$. By the central limit theorem $\sqrt{n}(\pmb{S}_{n}-\bt)\Rightarrow \mathcal{N}_{4}(\bzero,\bSig)$ where $\bSig=\mathrm{BlockDiag}(\rho _{1}\bSig_{1},\rho_{2}\bSig_{2})$ is a block diagonal matrix, $\rho _{i}=\lim (n/n_{i})$ for $i=1,2$. 
Under the null $P=Q$ so $\bSig_{1}=\bSig_{2}$ and therefore $\bSig$ can be estimated by $\bSig_{n}=\mathrm{
BlockDiag}(n/n_{1}\widehat{\bSig}_{0},n/n_{2}\widehat{\bSig}_{0})$ where 
\begin{equation*}
\widehat{\bSig}_{0}=
\left( 
\begin{array}{cc}
\frac{n_{11}+n_{21}}{n_{1}+n_{2}}(1-\frac{n_{11}+n_{21}}{n_{1}+n_{2}}) & 
\frac{n_{11}+n_{21}}{n_{1}+n_{2}}\frac{n_{13}+n_{23}}{n_{1}+n_{2}} \\ 
\frac{n_{11}+n_{21}}{n_{1}+n_{2}}\frac{n_{13}+n_{23}}{n_{1}+n_{2}} & \frac{n_{13}+n_{23}}{n_{1}+n_{2}}(1-\frac{n_{13}+n_{23}}{n_{1}+n_{2}})
\end{array}
\right).
\end{equation*}
Now $\pmb{\eta}=\pmb{R\theta }$ can be estimated by $\pmb{W}_{n}=\pmb{RS}_{n}$ and we find that 
\begin{equation*}
\pmb{W}_{n}^{[\mathrm{T5}]}=\left( 
\begin{array}{c}
0.20 \\ 
0.09
\end{array}
\right) ,\text{ }\pmb{W}_{n}^{[\mathrm{T6}]}=\left( 
\begin{array}{c}
-0.53 \\ 
0.21
\end{array}
\right) \text{ and }\boldsymbol{V}_{n}=\boldsymbol{R\Sigma }_{n}\boldsymbol{R}^{T}=\left( 
\begin{array}{cc}
0.75 & 0.16 \\ 
0.16 & 0.53
\end{array}
\right) .
\end{equation*}
Here $\pmb{W}_{n}^{[\mathrm{T5}]}$ and $\pmb{W}_{n}^{[\mathrm{T6}]}$ are the values of $\pmb{W}_{n}$ calculated from their Tables 5 and 6 respectively. Since the margins of both tables are equal the estimated variance under the null, i.e., $\boldsymbol{V}_{n}$, is the same in both settings. The p--values associated with their Table 5 are $(\alpha^{\ast},\gamma^{\ast}) =( 0.12,0.96)$. This finding indicates that the null hypothesis in \eqref{Eq.Hypotheses.PQ} can not be rejected in favor of the alternative in \eqref{Eq.Hypotheses.PQ}. A p--value of $0.12$ does indicate that there is some evidence, albeit weak, that $P\prec _{st}Q$. The p--value associated with the auxiliary hypothesis, i.e., with $H_{0}:P\preceq _{st}Q$ versus $H_{1}:P\nprec _{st}Q$ is $0.96$ indicating that the test is safe. Thus, it seems that their Table 5 suggests an ordering but is under--powered. In fact if we (artificially) double the number of observations in each cell in their Table 5 then we obtain the p--values $(\alpha^{\ast },\gamma ^{\ast}) =(0.03,0.96)$ indicating that the original null can be safely rejected. The p--values associated with their Table 6 are $(\alpha^{\ast},\gamma^{\ast})=(0.01,0.001)$. These p--values indicate that both the original null and the auxiliary null are rejected. In other word using $T_{n}$ would lead to the rejection of the original null, and very likely a Type III error, whereas $T_{n}^{\mathrm{SAFE}}$ would protect against such a potentially erroneous conclusion. It is worth mentioning that Cohen and Sackrowitz (1998) write that "some statisticians would be more inclined to assert stochastic order for Table 5 than for Table 6" a statement with which we fully agree and for which we now have a well--founded rigorous explanation.

\section{Summary and discussion\label{Sec.Summary}}

This paper develops a methodology for constructing safe tests with a focus on Type A Problems in ORI. We believe that by alleviating the problem of Type III errors, the proposed methodology addresses some of the principled objection to the use of ORI. We hope that this advance will allow statisticians and researchers working in a variety of application areas to capitalize on the benefits of ORI without the fear of systematic errors. 

The new testing procedure combines two base--tests; $T_n$ for the original system of hypotheses, i.e., \eqref{Eq.Hypotheses.TypeA}, and $T_n'$ for the auxiliary hypotheses \eqref{Eq.Hypotheses.TypeB}. The proposed approach can also be described as two--step procedure as indicated in Table \ref{Table.Decisions}. In Step One, the null hypothesis in \eqref{Eq.Hypotheses.TypeB} is tested at a predetermined level $\gamma$. If the null is rejected, the procedure terminates, and one concludes that neither the original null nor the alternative hypotheses are supported by the data and therefore implausible. In such cases, researchers are advised to reevaluate the available evidence, reconsider their modeling assumptions, and, collect additional data. If the auxilliary null is not rejected, a certificate of validity is issued, and it is appropriate to proceed to Step Two, in which the hypotheses in \eqref{Eq.Hypotheses.TypeA} are tested at the significance level $\alpha$ controlling the overall probability of a Type I error at the level $\alpha^{\mathrm{SAFE}}$.

The level of the composite safe test denoted by $\alpha^{\mathrm{SAFE}}$ is a function of the levels of the base tests $\alpha$ and $\gamma$. For any fixed value of $\gamma$ one can easily adjust $\alpha$ so $\alpha^{\mathrm{SAFE}}$ achieves any prescribed value in $(0,1)$. The value of $\gamma$ is immaterial in large samples provided $c_{\gamma}'/\sqrt{n}\rightarrow 0$ when $n\rightarrow \infty$. In finite samples, the value of $\gamma$ modulates the relationship between the power of the test and its Type III error rate, especially for values of $\bt$ near the boundary of the alternative. Since in Type A Problems $\mL\subset\mC$ it seems that choosing $\gamma<\alpha$ is coherent, i.e., from a purely logical perspective rejecting the auxiliary hypotheses should require stronger evidence than rejecting the original null. Nevertheless, we have also experimented with values $\gamma>\alpha$. These are sensible in situations where Type III errors may have a profound negative impact or the experimenter may have low confidence in the original formulation of the testing problem. We believe that an optimal choice of $\gamma$ can be made by formally balancing the power with the possibility of Type III errors. This is an open problem. 

The developments in this paper can be generalized in many directions. For example, safe tests for infinite dimensional hypotheses of the form $H_{0}:F(x)=G(x)$ versus $H_{1}:F(x)\gvertneqq G(x)$ 
where the symbol $\gvertneqq~$ indicates that a weak inequality holds for all $x\in \mathbb{R}$ and a strict inequality holds for some $x\in \mathbb{R}$ can be developed. These hypotheses can be tested by employing the ordinal dominance curve, (Davidov and Herman, 2012), or other reasonable alternatives. The auxiliary hypotheses are $H_{0}':F(x)\geq G(x)$ versus $H_{1}':\exists x\in \mathbb{R}$ such that $G(x)>F(x)$ for which an ordinal curved based test can also be developed. It is clear that a safe test is possible; its structure and properties require further study. 

The framework developed here can be extended to construct safe tests for general hypotheses of the form 
$H_{0}:\bt\in \Theta_{0}$ versus $H_{1}:\bt\in \Theta_{1}$ where $\Theta_{0}$ and $\Theta_{1}$ arbitrary subsets of $\Theta$. An interesting direction involves applying safe tests to Neyman--Pearson two--point hypotheses. Initial results suggest that the resulting procedures exhibit several appealing properties.

In summary, while order restrictions can substantially improve statistical efficiency and interpretability, these gains hinge on the validity of the assumed order. If the constraints are incorrect, inference may be seriously distorted. Consequently, there is a clear need for procedures that adaptively determine whether the data support an ordering and impose constraints only when warranted. This manuscript constitutes an initial contribution toward the development of such an adaptive methodology in ORI. Here hypothesis testing is addressed; extensions to constrained estimation, classification, and prediction forming the basis of future work.

\bigskip

\begin{center}
\underline{\textsc{Acknowledgments}}
\smallskip
\end{center}

The research of Ori Davidov was supported in part by the Israeli Science Foundation Grant No. 2200/22.

\begin{center}
\bigskip \medskip
\end{center}

\section{Appendix A: Proofs}

\textbf{Proof of Theorem \ref{Thm-Geom}:}

\smallskip

\begin{proof}
In Type A Problems $\Delta=\|\bt-\Pi_{\bSig}(\bt|\mL)\|_{\bSig}^{2}-\|\bt-\Pi_{\bSig}(\bt|\mC)\|_{\bSig}^{2}$ which by part $(h)$ of Proposition 3.12.6 in Silvapulle and Sen (2005) reduces to 
\begin{equation*}
\Delta =\|\Pi_{\bSig}(\bt|\mL)-\Pi_{\bSig}(\bt|\mC)\|_{\bSig}^{2}.
\end{equation*}
In addition by part $(i)$ of Proposition 3.12.6 of Silvapulle and Sen (2005) we may reexpress $\Delta$, given in the display above, as $\Delta =\|\Pi_{\bSig}(\bt|\mL^{\bot}\cap\mC)\|_{\bSig}^{2}$. It follows from Moreau's Theorem (Proposition 3.12.4 in Silvapulle and Sen, 2005) that $\Pi_{\bSig}(\bt|\mL^{\bot}\cap\mC)=\bzero$ if and only if $\bt\in(\mL^{\bot}\cap\mC)^{\circ}$. Thus we conclude that
in Type A Problems the DT is consistent provided
\begin{equation*}
\bt\notin (\mL^{\bot}\cap\mC)^{\circ}
\end{equation*}
establishing \eqref{Eq.Con.DT.TypeA}. Next, consider Type B Problems where $\Theta_{0}=\mathcal{C}$ and $\Theta_{1}=\mathcal{\mathbb{R}}^{m}$ in which case $\Delta$ reduces to $\|\bt-\Pi_{\bSig}(\bt|\mC)\|_{\bSig}^{2}$ which is strictly non--negative if and only if $\bt\notin\mC$. This completes the proof.  
\end{proof}

\medskip

\textbf{Proof of Theorem \ref{Thm-Acceptance.Regions}:}

\smallskip

\begin{proof}
Consider Type A Problems. By Theorem 3.7.1 in Silvapulle and Sen (2005) testing $H_{0}:\bt\in \mL$ versus $H_{1}:\bt\in \mC\backslash\mL$ is equivalent to
testing $H_{0}:\bt=\bzero$ versus $H_{1}:
\bt\in\mathcal{C\cap L}^{\bot}$. The DT for the latter system of hypotheses simplifies to
\begin{equation*}
T_{n}=n\{\|\pmb{S}_{n}\|_{\bSig_{n}}^{2}-\|\pmb{S}_{n}-\Pi_{\bSig_{n}}(\pmb{S}_{n}|(\mC\cap\mL^{\bot })_{\bSig_{n}}^{\circ})\|_{\bSig_{n}}^{2}\}
=n\|\Pi_{\bSig_{n}}(\pmb{S}_{n}|(\mC\cap\mL^{\bot})_{\bSig_{n}}^{\circ })\|_{\bSig_{n}}^{2}
\end{equation*}
where the second equality in the display above follows from Moreua's Theorem. Thus, the null is not rejected if $T_{n}<c_{\alpha }$, i.e., when
\begin{equation*}
\|\pmb{S}_{n}-(\mC\cap \mL^{\bot})_{\bSig_{n}}^{\circ}\|_{\bSig_{n}}^{2}<\frac{c_{\alpha }}{n}.
\end{equation*}
It is easy to verify that the latter is satisfied if and only if $
\pmb{S}_{n}\in\mathcal{A}_{n}(\mL,\mC,\alpha) $ where 
\begin{equation} \label{Eq.Pf.Thm.Acc.2}
\mathcal{A}_{n}(\mL,\mC,\alpha) =(\mC\cap\mL^{\bot})_{\bSig_{n}}^{\circ}\oplus\mathrm{Ball}_{\bSig_{n}}(\bzero,\sqrt{\frac{c_{\alpha }}{n}})
\end{equation}
establishing \eqref{Eq.1.Thm.Acc} as required. Next we consider a Type B Problem in which the hypotheses $H_{0}:\bt\in\mC$ versus $H_{1}:\bt\in \mathbb{R}^{m}\backslash\mC$ are tested. Clearly, the DT reduces to
\begin{equation*}
T_{n}=n\{\|\pmb{S}_{n}-\Pi_{\bSig_{n}}(\pmb{S}_{n}|\mC)\|_{\bSig_{n}}^{2}=n\|\Pi_{\bSig_{n}}(\pmb{S}_{n}|\mC)\|_{\bSig_{n}}^{2}.
\end{equation*}
Therefore $T_{n}<c_{\alpha}$ if and only if 
\begin{equation*}
\|\pmb{S}_{n}-\mC\|_{\bSig_{n}}^{2}<\frac{c_{\alpha }}{n}.
\end{equation*}
Repeating the calculations above, we find that 
\begin{equation*}
\mathcal{A}_{n}(\mC,\mathbb{R}^{m})=\mathcal{C}\oplus \mathrm{Ball}_{\bSig_{n}}(\bzero,\sqrt{\frac{c_{\alpha}}{n}})
\end{equation*}
concluding the proof.
\end{proof}

\medskip

\textbf{Proof of Theorem \ref{Thm-3}:}

\smallskip

\begin{proof}
It is well known (e.g., Barlow et al. 1972) that the DT for testing $H_{0}:\bt\in\mL$ versus $H_{1}:\bt\in \mC\backslash \mL$, where $\mL=\{\bt:\theta_{1}=\cdots =\theta_{K}\}$ and $\mC$ is an any closed convex cone is of the form
\begin{equation*}
T_{n}=\frac{n}{\widehat{\sigma}_{n}^{2}}\sum_{i=1}^{K}w_{i,n}(\widetilde{\theta}_{i,n}-\overline{Y}_{n})^{2},
\end{equation*}
where $\overline{Y}_{n}=n^{-1}\sum_{i=1}^{K}\sum_{j=}^{n_{i}}Y_{ij}$ is the overall mean, $(w_{1,n},\ldots ,w_{K,n})=(n_{1}/n,\ldots ,n_{K}/n)$ is a vector of weights determined by the sample sizes, and $\widetilde{\theta}_{1,n},\ldots,\widetilde{\theta}_{K,n}$ are the elements of $\widetilde{\bt_{n}}=\Pi_{\bSig_{n}}(\pmb{S}_{n}|\mC)$ the constrained estimator of $\bt$. Here $\pmb{S}_{n}=(\overline{Y}_{1,n},\ldots ,\overline{Y}_{K,n})^{T}$ with $\overline{Y}_{i,n}=n_{i}^{-1}\sum_{j=}^{n_{i}}Y_{ij}$ for $i=1,\ldots ,K$ and $\bSig_{n}=\widehat{\sigma}_{n}^{2}\mathrm{Diag}(1/w_{1,n},\ldots,1/w_{K,n})$ where $\widehat{\sigma}_{n}^{2}$ is any consistent estimator for $\sigma^{2}$. The resulting estimators are known as the isotonic regression estimators. In particular if $\mathcal{C=C}_{s}$ the $i^{th}$ element of $\widetilde{\bt}_{n}$ is given by the min--max formula 
\begin{equation} \label{Eq.maxmin.Thm3}
\widetilde{\theta }_{i,n}=\min_{t\geq i}\max_{s\leq i}\mathrm{Av}(\pmb{S}_{n},s,t).  
\end{equation}
Since $\sqrt{n}(\pmb{S}_{n}-\bt)\Rightarrow \mathcal{N}_{K}(\bzero,\bSig)$ and $\bSig_{n}\overset{p}{\rightarrow}\bSig=\sigma^{2}\mathrm{Diag}(1/w_{1},\ldots,1/w_{K})$  as $n\rightarrow\infty$ it is evident that for large $n$ we have
\begin{equation*}
\frac{1}{\widehat{\sigma}_{n}^{2}}\sum_{i=1}^{K}w_{i,n}(\widetilde{\theta}_{i,n}-\overline{Y}_{n})^{2}=\frac{1}{\sigma^{2}}\{\sum_{i=1}^{K}w_{i}(\theta_{i}^{\ast}-\overline{\theta})^{2}\}+o_{p}(1)
\end{equation*}
where $\overline{\theta}=\sum_{i=1}^{K}w_{i}\theta _{i}$ is the almost sure limit of $\overline{Y}_{n}$ and $\theta _{i}^{\ast },$ $i=1,\ldots,K$ are the limiting value of $\widetilde{\theta}_{n,i}$ obtained by applying the law of large numbers and the continuous mapping theorem to \eqref{Eq.maxmin.Thm3}, i.e., 
\begin{equation*}
\theta _{i}^{\ast}=\min_{t\geq i}\max_{s\leq i}\mathrm{Av}\left( 
\boldsymbol{\theta },s,t\right).
\end{equation*}
It follows that the DT is consistent if and only if
\begin{equation} \label{Eq.1.Thm3}
\sum_{i=1}^{K}w_{i}(\theta_{i}^{\ast }-\overline{\theta })^{2}>0.
\end{equation}
It is well known that $\overline{Y}_{n}=\sum_{i=1}^{K}w_{i,n}\widetilde{\theta}_{i,n}$ so $\overline{\theta}=\sum_{i=1}^{K}w_{i}\theta_{i}=\sum_{i=1}^{K}w_{i}\theta_{i}^{\ast}$. Now by construction $\theta_{1}^{\ast}\leq \cdots \leq \theta_{K}^{\ast}$ and $\overline{\theta }=\sum_{i=1}^{K}w_{i}\theta_{i}^{\ast}$ and therefore \eqref{Eq.1.Thm3} holds if and only if for some $i$ we have $\theta _{i}^{\ast}<\theta_{i+1}^{\ast}$.

Recall that the min--max formulas are equivalent to the celebrated pool--adjacent violator algorithm (PAVA) as described in van Eeden (2006). By the loop--invariant property of PAVA we can first apply PAVA separately to the vectors $\pmb{S}_{n}(1,i)=(\overline{Y}_{1,n},\ldots ,\overline{Y}_{i,n})^{T}$ and $\pmb{S}_{n}(i+1,K)=(\overline{Y}_{i+1,n},\ldots ,\overline{Y}_{K,n})^{T}$ obtain the estimators $\widetilde{\bt}_{n}(1,i)$ and $\widetilde{\bt}_{n}(i+1,K)$ and then apply PAVA again to the full vector $(\widetilde{\bt}_{n}(1,i)^{T},\widetilde{\bt}_{n}(i+1,K)^{T})$ to obtain $\widetilde{\bt}_{n}$. Furthernote that by the min--max formulas applied separately to the vectors $\pmb{
S}_{n}(1,i)$ and $\pmb{S}_{n}(i+1,K)$ we find that 
\begin{eqnarray*}
\widetilde{\theta }_{i,n}(1,i) &=&\max_{1\leq s\leq i}\mathrm{Av}(
\pmb{S}_{n}(1,i),s,i), \\
\widetilde{\theta }_{i+1,n}(i,K+1) &=&\min_{i+1\leq t\leq K}\mathrm{Av}(\pmb{S}_{n}\left( i+1,K\right) ,i+1,t).
\end{eqnarray*}
Here $\widetilde{\theta}_{i,n}(1,i)$ is the $i^{th}$ and largest element of $\widetilde{\bt}_{n}(1,i)$ whereas $\widetilde{\theta}_{i+1,n}(i,K+1)$ is the $1^{st}$ and smallest element of $\widetilde{\bt}_{n}(i+1,K)$. By the law of large numbers and the continuous mapping theorem $\widetilde{\theta}_{i,n}(1,i)\overset{p}{\rightarrow}\theta _{i}^{\ast}(1,i)$ and $\widetilde{\theta}
_{i+1,n}(i,K+1)\overset{p}{\rightarrow}\theta _{i+1}^{\ast}(i+1,K)$ where 
\begin{eqnarray*}
\theta _{i}^{\ast}(1,i) &=&\max_{1\leq s\leq i}\mathrm{Av}(\bt( 1,i),s,i), \\
\theta _{i+1}^{\ast}(i+1,K) &=&\min_{i+1\leq t\leq K}\mathrm{Av}(
\bt(i+1,K,i+1,t).
\end{eqnarray*}
Now if
\begin{equation} \label{Eq.2.Thm3}
\theta _{i}^{\ast}(1,i)<\theta _{i+1}^{\ast}(i+1,K)  
\end{equation}
then $\bt^{\ast}=(\bt^{\ast}(1,i)^{T},\bt^{\ast}(i+1,K)^{T})^{T}$, i.e., applying PAVA (or min--max) to $\pmb{S}_{n}$ is equivalent to separately applying it to $\pmb{S}_{n}(1,i)$ and $\pmb{S}_{n}(i+1,K)$ and then combining the results. Thus, if \eqref{Eq.2.Thm3} holds for some $i\in\{1,\ldots,K\}$ then $\theta _{i}^{\ast}<\theta_{i+1}^{\ast}$ and consequently \eqref{Eq.1.Thm3} holds as well so the DT is consistent. However, Equation \eqref{Eq.2.Thm3} is nothing but Equation \eqref{Eq.Thm3}. This completes the proof.
\end{proof}

\medskip

\textbf{Proof of Theorem \ref{Thm-DT.not.SAFE}:}

\smallskip

\begin{proof}
Consider Type A Problems. Clearly for any $\bt\in \mC$ we have $\Delta >0$, so the DT is consistent on $\mC$. By Equation \eqref{Eq.Con.DT.TypeA} the DT is not consistent on $(\mC\cap\mL^{\bot})^{\circ}$. First note that
\begin{equation*}
(\mC\cup (\mC\cap \mL^{\bot})^{\circ })\subseteq 
\mC\cup \mC^{\circ}.
\end{equation*}
Next note that
\begin{equation*}
\mathbb{R}^{m}\backslash (\mC\cup \mC^{\circ})=\{\pmb{u}\in \mathbb{R}^{m}:\pmb{u}^{T}\bSig^{-1}\pmb{v}>0,\forall 
\pmb{v}\in \mC\}\neq \varnothing,
\end{equation*}
i.e., $\mC\cup\mC^{\circ}$ is a strict subset of $\mathbb{R}^{m}$. It follows that DT is consistent on $\bt\in \mathbb{R}^{m}\backslash (\mC\cup (\mC\cap\mL^{\bot})^{\circ})$. In other words there are $\bt\notin \mC$ for which $\lim_{n}\mathbb{P}_{\bt}(T_{n}\in \mathcal{R})=1$. Thus the DT is not safe. 
\end{proof}

\medskip

\textbf{Proof of Theorem \ref{Thm-4}:}

\smallskip

\begin{proof} Fix $\alpha$ and let $c_{\alpha}=\sup_{\bt\in\mL}\mathbb{P}_{\bt}(T_{n}\geq c_{\alpha})$. Now by construction $T_{n}^{\mathrm{SAFE}} \le T_{n}$ and therefore $\{T_{n}^{\mathrm{SAFE}}\geq c_{\alpha}\} \subseteq \{T_{n}\geq c_{\alpha}\}$. It immediately follows that 
\begin{align*}
\alpha^{\mathrm{SAFE}} =\sup_{\bt\in\mL}\mathbb{P}_{\bt}(T_{n}^{\mathrm{SAFE}}\geq c_{\alpha})\le \sup_{\bt\in\mL}\mathbb{P}_{\bt}(T_{n}\geq c_{\alpha})=\alpha, 
\end{align*}
establishing \eqref{Eq.1.Thm.4}.  

Next let $\mathcal{E}$ denote the set of values of $\pmb{\theta}$ on which the test $T_{n}$ commits a Type III error with probability one as $n\rightarrow \infty$, i.e., 
\begin{equation*}
\mathcal{E}=\mathbb{R}^{p} \setminus (\mathcal{C} \cup \mathcal{C}^{\circ }),
\end{equation*}
cf., Theorem \ref{Thm-DT.not.SAFE}. Now for any $\pmb{\theta}\in \mathcal{E}$ and $\varepsilon>0$ we have
\begin{eqnarray*}
\mathbb{P}_{\pmb{\theta}}(T_{n}^{\mathrm{SAFE}} 
&\geq& c_{\alpha}) = \mathbb{P}_{\pmb{\theta}}(T_{n} \geq c_{\alpha}, T_{n}^{'} < c_{\gamma}') \\
&\leq& \mathbb{P}_{\pmb{\theta}}(T_{n}^{'} \le c_{\gamma}') = \mathbb{P}_{\pmb{\theta}}(\pmb{S}_{n} \in 
\mathcal{A}_{n}(\mathcal{C},\mathcal{\mathbb{R}}^{p},\gamma)) \\
&=& \mathbb{P}_{\pmb{\theta}}(\pmb{S}_{n} \in \mathcal{C} \oplus \mathrm{Ball}_{\pmb{\Sigma}_n}(\mathbf{0}, \sqrt{\frac{c_{\gamma}'}{n}})) \\
&=& \mathbb{P}_{\pmb{\theta}}(\pmb{S}_{n} \in \mathcal{C} \oplus \mathrm{Ball}_{\pmb{\Sigma}}(\mathbf{0}, \sqrt{\frac{c_{\gamma}'}{n}})) + o_p(1) \\
&=& \mathbb{P}_{\pmb{\theta}}(\pmb{S}_{n} \in 
\mathcal{C} \oplus \mathrm{Ball}_{\pmb{\Sigma}}(\mathbf{0}, \sqrt{\frac{c_{\gamma}'}{n}}), \pmb{S}_{n} \in \mathrm{Ball}_{\pmb{\Sigma}}(\pmb{\theta}, \varepsilon)) \\
&+& \mathbb{P}_{\pmb{\theta}}(\pmb{S}_{n} \in 
\mathcal{C} \oplus \mathrm{Ball}_{\pmb{\Sigma}}(\mathbf{0}, \sqrt{\frac{c_{\gamma}'}{n}}), \pmb{S}_{n} \notin \mathrm{Ball}_{\pmb{\Sigma}}(\pmb{\theta}, \varepsilon)) + o_p(1) \\
&\leq& \mathbb{P}_{\pmb{\theta}}(\pmb{S}_{n} \in (\mathcal{C} \oplus \mathrm{Ball}_{\pmb{\Sigma}}(\mathbf{0}, \sqrt{\frac{c_{\gamma}'}{n}}) \cap \mathrm{Ball}_{\pmb{\Sigma}}(\pmb{\theta}, \varepsilon))) \\
&+& \mathbb{P}_{\pmb{\theta}}(\pmb{S}_{n} \notin \mathrm{Ball}_{\pmb{\Sigma}}(\pmb{\theta}, \varepsilon)).
\end{eqnarray*}
Since $\sqrt{n}(\pmb{S}_{n}-\pmb{\theta})\Rightarrow \mathcal{N}_{p}(\bzero,\pmb{\Sigma})$ we have 
\begin{equation*}
\mathbb{P}_{\pmb{\theta}}(\pmb{S}_{n}\notin \mathrm{Ball}_{
\pmb{\Sigma}}(\pmb{\theta},\varepsilon))\rightarrow 0
\end{equation*}
as $n\rightarrow \infty$. Now $\pmb{\theta}\in\mathcal{E}$ so for
all small enough $\varepsilon$ we have $\mathrm{Ball}_{\pmb{\Sigma}}(\pmb{\theta},\varepsilon )\cap \mathcal{C}_{+}^{p}=\varnothing$. Consequently, for large enough $n$ we have $\mathcal{C}\oplus\mathrm{Ball}_{\pmb{\Sigma}}(\bzero,\sqrt{c_{\gamma }'/n})\cap \mathrm{Ball}_{\pmb{\Sigma}}(\pmb{\theta},\varepsilon )=\varnothing$. Hence, 
\begin{equation*}
\mathbb{P}_{\pmb{\theta}}(\pmb{S}_{n}\in (\mathcal{C}\oplus \mathrm{Ball}_{\pmb{\Sigma}}(\bzero,\sqrt{\frac{c_{\gamma}'}{n}})\cap\mathrm{Ball}_{\pmb{\Sigma}}(\pmb{\theta},\varepsilon ))=0
\end{equation*}
for all large $n$. Combining the displays above we conclude that $\mathbb{P}_{\pmb{\theta}}(T_{n}^{\mathrm{SAFE}}\geq c_{\alpha })\rightarrow 0$ and therefore $T_{n}^{\mathrm{SAFE}}$ is a safe test as claimed.

Finally, since for any $\alpha \in (0,1)$ we have $\alpha ^{\mathrm{SAFE}}<\alpha$, i.e., $\mathbb{P}(T_{n}^{\mathrm{SAFE}}\ge c_{\alpha})<\alpha$. Additionally, there must exist a value $c_{\alpha }^{\mathrm{SAFE}}$ for which and $\sup_{\bt\in\mL}\mathbb{P}(T_{n}^{\mathrm{SAFE}}\ge c_{\alpha }^{\mathrm{SAFE}})=\alpha$. Thus, $c_{\alpha }^{\mathrm{SAFE}}<c_{\alpha }$. Let $\bt\in \mathrm{int}(\mC\backslash\mL)$. For any such $\pmb{\theta}$ we have $\mathbb{P}_{\pmb{\theta}}(T_{n}^{'}<c_{\gamma}')\rightarrow 1$ as $n\rightarrow \infty $ and since $T_{n}^{\mathrm{SAFE}}=T_{n}\mathbb{I}_{\{T_{n}^{'}<c_{\gamma}'\}}$ we have 
\begin{equation*}
T_{n}^{\mathrm{SAFE}}=T_{n}+o_{p}(1).
\end{equation*}
Moreover, for such $\bt$ we have
\begin{equation*}
\mathbb{P}_{\bt}(T_{n}^{\mathrm{SAFE}}\geq c_{\alpha })=\mathbb{P}_{\pmb{\theta}}(T_{n}\geq c_{\alpha })+o_{p}(1) <\mathbb{P}_{\pmb{\theta}}(T_{n}\geq c_{\alpha}^{\mathrm{SAFE}})
\end{equation*}
establishing \eqref{Eq.2.Thm.4}.
\end{proof}

\end{document}